\documentclass[12pt]{article}
\usepackage{amsthm,amsfonts,amsmath, amscd, bbold, bm}
\usepackage{mathrsfs}
\usepackage{accents}
\usepackage{color}
\usepackage{tikz}

\usepackage{mathrsfs}
\usepackage{amscd}
\usepackage[active]{srcltx}
\usepackage{verbatim}
\newcommand{\version}{June 4, 2018}

\usepackage{amsthm,amsfonts,amsmath, amscd}
\usepackage{amssymb,amsmath, dsfont}
\usepackage{graphicx}
\swapnumbers
\theoremstyle{plain}

\newtheorem{thm}{THEOREM}[section]
\newtheorem{lm}[thm]{LEMMA}
\newtheorem{cl}[thm]{COROLLARY}

\newtheorem{exam}[thm]{EXAMPLE}
\theoremstyle{definition}
\newtheorem{defi}[thm]{DEFINITION}
\theoremstyle{definition}
\newtheorem{remark}[thm]{Remark}

\newcommand{\cH}{{\mathcal{H} }}

\newcommand{\one}{{{\bf 1}}}

\newcommand{\cA}{{\mathord{\cal A}}}

\newcommand{\dd}{{\, \rm d}}
\newcommand{\tr}{{\rm Tr}}

\renewcommand{\|}{{\Vert}}

\numberwithin{equation}{section}
\def\H{\mathcal{H}}
\def\K{\mathcal{K}}
\def\fH{\mathfrak{H}}

\def\dd{{\rm d}}

\def\SS{\mathcal{S}}

\def\ncht{\left(\begin{matrix} N\cr 2\cr \end{matrix}\right)}

\newcommand{\cP}{{\mathord{\mathscr P}}}
\newcommand{\cQ}{{\mathord{\mathscr Q}}}

\newcommand{\cL}{{\mathord{\mathscr L}}}
\newcommand{\cK}{{\mathord{\mathscr K}}}
\newcommand{\cB}{{\mathord{\mathscr B}}}
\newcommand{\cU}{{\mathord{\mathscr U}}}
\newcommand{\cT}{{\mathord{\mathscr T}}}
\newcommand{\cC}{{\mathord{\mathscr C}}}
\newcommand{\cX}{{\mathord{\mathscr X}}}

\newcommand{\cD}{{\mathord{\mathscr D}}}

\newcommand{\cJ}{{\mathcal{J}}}
\newcommand{\aC}{{\mathcal{C}}}

\newcommand{\aK}{{\mathcal{K}}}

\newcommand{\bma}{{\bm{\alpha}}}
\newcommand{\bmb}{{\bm{\beta}}}

\newcommand{\R}{{\mathord{\mathbb R}}}

\newcommand{\C}{{\mathord{\mathbb C}}}

\newcommand{\N}{{\mathord{\mathbb N}}}
\newcommand{\spec}{{\rm{Spec}}(H_N)}
\newcommand{\Dens}{{\mathfrak S}}

\newcommand{\EC}{{\rm{E}}_{\cC_N}}
\newcommand{\bk}{\rangle \langle }

\begin{document}


\def\tr{{\rm Tr}}

\title{Chaos, Ergodicity, and Equilibria in a  Quantum  Kac Model }

\author{\vspace{5pt} Eric A. Carlen$^{1,2}$, Maria C. Carvalho$^{1,2}$   and
Michael P. Loss$^{2}$ \\
\vspace{5pt}\small{$1.$ Department of Mathematics, Hill Center,}\\[-6pt]
\small{Rutgers University,
110 Frelinghuysen Road
Piscataway NJ 08854-8019 USA}\\
\vspace{5pt}\small{$2.$ CMAF-CIO, University of Lisbon, P 1749-016 Lisbon, Portugal}\\
\vspace{5pt}\small{$3.$ School of Mathematics,
Georgia Tech } \\[-6pt]
\small{Atlanta GA 80332
  }\\
 }
\date{\version}
\maketitle 

\tableofcontents

\let\thefootnote\relax\footnote{
\copyright \, 2018 by the authors. This paper may be
reproduced, in its
entirety, for non-commercial purposes.}

\begin{abstract}
We introduce  quantum versions of the Kac Master Equation and the Kac Boltzmann Equation. We study the steady states of each of these 
equations, and prove a propagation of chaos theorem that relates them.   The Quantum Kac Master Equation (QKME) describes a quantum 
Markov semigroup 
$\cP_{N,t}$, while the Kac Boltzmann Equation describes a non-linear evolution of density matrices on the single particle state space.   
All of the steady states of the $N$ particle quantum system described by the QKME are separable, and thus the evolution 
described by the QKME is entanglement breaking. The results set the stage for a quantitative study of approach to equilibrium in quantum 
kinetic theory, and a quantitative study of  the rate of destruction of entanglement in a class of quantum Markov semigroups describing binary interactions.

\end{abstract}

\medskip
\leftline{\footnotesize{\qquad Mathematics subject
classification numbers: 35Q20, 47L90, 46L57}}
\leftline{\footnotesize{\qquad Quantum Master Equation, Completely Positive, 
Equilibrium}}

\section{Introduction} \label{intro}

\subsection{The classical Kac model}

We begin by briefly recalling some essential features of the classical Kac Master Equation \cite{K56} which models a system of 
$N$ particles of mass $m$ with  one-dimensional velocities $v_j$, $j=1,\dots,N$ in $\R$ that   
interact only through binary collisions. At each instant any particle may collide with any other. 
That is, the Kac model is a mean-field model of binary molecular collisions  in one region of physical space.  
Between collisions there is no interaction between the particles, and the binary collisions conserve energy. Consequently, the total energy of the 
system is the sum of the individual kinetic energies:
\begin{equation}\label{clkac1}
E(v_1,\dots,v_N) = \frac{m}{2} \sum_{j=1}^N |v_j|^2\ .
\end{equation}

When a binary collision occurs, for some pair $(i,j)$, $1\leq i < j \leq N$,
$v_i$ and $v_j$ change to $v_i^*$ and $v_j^*$ respectively,  while for all other $k\in \{1,\dots,N\}$, $v_k$ is unchanged. 
Since the energy is conserved, $v_i^2+ v_j^2 = v_i^{*2}+ v_j^{*2}$. It follows that there is a one parameter family of kinematically possible collisions between particles $i$ and $j$. The vector $R_{i,j,\theta}\vec v$  of 
post-collisional velocities is related to the vector of  pre-collisional  velocities $\vec v$ by
\begin{equation}\label{clkac2}
(R_{i,j,\theta}\vec v)_k  = \begin{cases}  v_k & k \neq i,j\\
 \cos \theta v_i - \sin \theta v_j & k= i\\
 \sin \theta v_i  + \cos v_j & k=j
\end{cases}
\end{equation}
where $\theta\in [-\pi,\pi]$ is the {\em collision parameter}. 
For $E>0$, define $\SS_E := \{\vec v\ :\  E(\vec v) = E\}$.  Note that for all $1\leq i< j \leq N$ and all $\theta\in [-\pi,\pi]$, $R_{i,j,\theta}$ is an invertible transformation from  $\SS_E$ onto $\SS_E$; which is to say that the collisions 
conserve energy and are reversible. 

The {\em Kac Walk} on $\SS_E$ is the Markov jump process on $\SS_E$ in which at each step, a pair 
$(i,j)$, $1\leq i < j \leq N$, is chosen uniformly at random, along with a $\theta\in [-\pi,\pi]$, also chosen uniformly at random, and then the system jumps from the state $\vec v\in \SS_E$ prior to the collision to the new post-collisional state $R_{i,j,\theta}\vec v$.

The Kac Master Equation  is the evolution equation describing the continuous time version of the Kac Walk in which the jump times arrive in a Poisson stream with mean waiting time $1/N$, so that the mean waiting time for collisions involving any particular particle is independent of $N$.


%
%

Now fix the energy to be $N$, so that the average energy per particle is $1$, independent of $N$. Let 
$\cL_N$ denote the generator of this process on $\SS_N$ so that if the initial probability density for $\vec v$
is $F$, the density at time $t$ is $e^{t\cL_N}F$, where
\begin{equation}\label{clkac3}
\cL_N F(\vec v) = {N}{\ncht}^{-1}\sum_{i<j} (F^{i,j}-F)\quad{\rm and}\quad 
F^{i,j}(\vec v) := \frac{1}{2\pi}\int_{-\pi}^\pi F(R_{i,j,\theta}\vec v)\dd \theta\ .
\end{equation}
The equation 
\begin{equation}\label{clkac4}
\frac{\partial}{\partial t}F(\vec v,t) = \cL_N F(\vec v,t)\ ,
\end{equation}
which describes the evolution of probability densities for $\vec v$  on $\SS_N$ with respect to the uniform 
probability measure $\sigma_N$ on 
$\SS_N$, is the {\em Kac Master Equation} (KME). It is the Kolmogorov forward equation for the continuous Kac walk on 
$\SS_N$.

One motivation for investigating the KME is that it has a rigorous connection with the non-linear 
{\em Kac-Boltzmann equation}, a one dimensional caricature of the spatially homogeneous Boltzmann equation; see 
(\ref{KBE}) below. Moreover, a more complicated version of the Kac Master equation describing random collisions that 
preserve both the energy and the momentum of particles with velocities in $\R^3$ is connected in the same way with the 
usual spatially homogenous Boltzmann equation \cite{K59}. For simplicity, we only discuss Kac's  original one dimensional caricature, 
which is perfectly adequate to set the stage for the quantum mechanical investigation.  We refer to \cite{CCL3} for a recent survey on the Kac model that focuses on questions that are relevant to the present work. 

The connection between the Kac model and the Boltzmann equation depends on Kac's notion of 
{\em chaos} and its propagation under the stochastic evolution. 

\begin{defi}[chaos]
Let $\mu$ be a probability measure on $\R$. A sequence $\{\mu^{(N)}\}_{N\in \N}$ of 
probability measures on $\SS_N$ is $\mu$ chaotic in case for each $k\in \N$,
\begin{equation*}
\lim_{N\to\infty} \int \chi(v_1, \dots, v_k){\rm d} \mu^{(N)}(v_1, \dots, v_N) = 
\int \chi(v_1, \dots, v_k) {\rm d}\mu(v_1) \cdots {\rm d} \mu(v_k)
\end{equation*}
for each bounded continuous function $\chi$ on $\R^k$. 

\end{defi}

In 1956 Kac proved \cite{K56}:

\begin{thm}[Propagation of Chaos]\label{pch}
Let $\{ F_0^{(N)} \sigma^{(N)}\}$ be a $f_0(v) {\rm d} v$--chaotic sequence. Then 
$\{ e^{t\mathcal{L}_N}F_0^{(N)} \sigma^{(N)}\}$ is a $f(v, t) {\rm d} v$--chaotic sequence and 
$f(v,t)$ is the solution of the initial value problem
\begin{equation}\label{KBE}
{\partial \over \partial t}f(v,t)  = \frac{1}{\pi} \int_{-\pi}^\pi  
 \left(\int_\R [f( v^*(\theta,t) f(w^*(\theta),t)- f(v,t)f(w,t)]  {\rm d} w\right) {\rm d} \theta 
\end{equation}
with $f(v,0)=f_0(v)$, 
where $v^*(\theta)$ and $w^*(\theta)$ are given by
rotating $(v,w)$ through an angle $\theta$:
$$v^*(\theta) = \cos\theta v + \sin\theta w \quad{\rm and}\quad w^*(\theta) = -\sin\theta v + \cos\theta w\ .$$
\end{thm}

The equation in (\ref{KBE}) is the {\em Kac Boltzmann Equation} (KBE).
Theorem\ref{pch} would be of little relevance to the investigation of the Boltzmann equation were it not possible to construct $f(v){\rm d}v$-chaotic sequences for all physically relevant initial data for the Kac-Boltzmann Equation.
The following theorem shows that a construction proposed by Kac works in sufficient generality that this is indeed possible \cite{CCRLV}:

\begin{thm}[Existence of Chaotic initial data] \label{CCLRV}
Let $f$ be a probability density on $\R$ satisfying, for some $p>1$, 
$$
\int_\R f(v) v^2 {\rm d}v =1 \ , \ \int_\R f(v) v^4 {\rm d} v < \infty \ , f \in L^p(\R) 
$$
and let $\mu({\rm d}v) = f(v){\rm d} v$, and let
$[\mu^{\otimes N}]_{S^{N-1}(\sqrt N)}$ be the normalized restriction of 
$f^{\otimes N}$ 
to the sphere $\SS_N$. Then $\{ [\mu^{\otimes N}]_{S^{N-1}(\sqrt N)}\} $ is $\mu$ --chaotic.
\end{thm}

Now that the relation between the KME and the KBE is clarified, we turn to the question that motivated Kac: 
Can one obtain information on the rate of relaxation to equilibrium for the KME that is useful for proving results about the rate of relaxation to equilibrium of solutions of the KBE?  

It is easy to see that under the Kac Walk, the density tends to become uniform:  For all $F\in L^2(\SS_N,\sigma_N)$.
$$\lim_{t\to\infty} e^{t\cL_N} F = 1\ .$$
This is a simple consequence of the ergodicity of the Kac Walk:  It is possible to move from any point in 
$\SS_N$ to any other in at most $N-1$ steps of the Kac Walk.  However, this sort of argument says 
nothing about the rate of convergence. The rate of convergence in $L^2(\SS_N,\sigma_N)$ is
governed by the {\em spectral gap} of 
$\cL_N$, $\Delta_N$,  which is
\begin{equation}\label{gapdef}
\Delta_{N} := \inf\{ -\langle F,\cL_N F\rangle_{L^2(\SS_N)} \ :\ \|F\|_2 = 1\ , \langle F,1\rangle_{L^2(\SS_N)}  = 0\ \}\ .
\end{equation}
In his 1956 paper, Kac conjectured  that 
$\liminf_{N\to\infty}\Delta_{N}  > 0$. This was proved by Janvresse \cite{J01} with no estimate on the limiting gap. Our paper \cite{CCL} gave the exact value:
$$\Delta_N = \frac12 \frac{N+2}{N-1}\ .$$
Later, Maslin \cite{Mas} was able to compute many more eigenvalues, but their large multiplicty 
is such that this does not provide much further information in rates of convergence. 

 For the purpose of investigating the properties of the KBE, the relative entropy would provide a better measure, 
 and subsequent research will address spectral gaps and rates of entropy dissipation in quantum Kac models. In the 
 present paper we focus on quantum analogs of Kac's original 1956 result relating the Kac Master Equations and the 
 Kac Boltzmann Equation,
and in addition carry out a detailed study of the steady states of their quantum analogs.   

\subsection{Quantum Markov semigroups}

In the next subsection we describe  the quantum analog of the Kac Master Equation that is investigated here. 
In the quantum setting, probability densities are replaced by density matrices; that is, positive trace class operators on a 
Hilbert space that have unit trace. The 
quantum analog \cite{Dopen} of the semigroup $e^{t\cL_N}$ arising in the continuous time Kac Walk will be a particular sort 
of  evolution equation for density matrices known as a {\em quantum Markov semigroup}. Before going into the particular 
features of the quantum model, we introduce the notation that we will use. 

Let $\fH$ be a separable Hilbert space with inner product $\langle \cdot, \cdot \rangle_\fH$ and norm $\|\cdot\|_\fH$.
Let $\cB(\fH)$ denote the set of norm-continuous linear operators on $\fH$, and for $A\in \cB(\fH)$, let $\|A\|_\infty$
denote the operator norm of $A$.  We make use of other norms on subspaces of $\cB(\fH)$.  For $p\in [1,\infty)$,   $\cT_p(\fH)$ denotes the set of operators $A$ on $\fH$ such that
$$\tr[(A^*A)^{p/2}] < \infty $$
equipped with the norm
$$\|A\|_p = (\tr[(A^*A)^{p/2}])^{1/p}\ .$$
For each $p\in [1,\infty)$,  $ \cT_p(\fH)$ is a two-sided ideal in $\cB(\fH)$, closed in the $\|\cdot\|_p$ norm, but operator norm dense in $\cC(\fH)$, the norm closed subalgebra of compact operators on $\fH$. 

Of special interest to us are  $\cT_1(\fH)$, the space of {\em trace class} operators on $\fH$, and  
$\cT_2(\fH)$, which is a Hilbert space in its own right with the inner product $\langle A,B\rangle_{\cT_2(\fH)} = 
\tr[A^*B]$. 

If $A \in  \cT_1(\fH)$, then $B\mapsto \tr[AB]$ is a bounded linear functional on $\cB(\fH)$, and every bounded linear functional on $\cB(\fH)$ has this form. That is, $ \cT_1(\fH)$ is the predual of  $\cB(\fH)$.
In the same way, $\cT_1(\fH)$
itself is the dual of $\cC(\fH)$.

Define $\Dens(\fH)$ to be the set of positive operators $\rho$  in $\cT_1(\fH)$ such that $\tr[\rho] =1$.  We refer to 
$\Dens(\fH)$ as the set of {\em density matrices} on $\fH$. 

%

A  linear  operator $\cK : \cB(\fH) \to \cB(\fH)$ is {\em positivity preserving} in case $\cK A \geq 0$ 
whenever $A\geq 0$. It is {\em completely positive} in case the following more stringent condition is satisfied: 
For each $n\in \N$, let $M_n(\C)$ denote the space of complex $n\times n$ matrices. Let 
$E_{i,j}$ denote the element of 
$M_n(\C)$ that has $1$ in the $i,j$ entry, and $0$ elsewhere. The set 
$\{E_{i,j}\}_{1,\leq i,j\leq n}$ is the {\em matrix unit} basis for $M_n(\C)$. 
The general element of $\cB(\fH)\otimes M_n(\C)$ can be written as a sum
$$\sum_{i,j=1}^n A_{i,j} \otimes E_{i,j}\ ,$$
and may therefore be regarded as an $n\times n$ matrix with entries in $\cB(\fH)$, and consequently, as an operator
on $\oplus^n\fH$, the direct sum of $n$ copies of $\fH$. We say that an element of 
$\cB(\fH)\otimes M_n(\C)$ is positive if it is positive as an operator on  $\oplus^n\fH$.  
Any linear transformation $\cK : \cB(\fH)\to \cB(\fH)$ induces the  linear transformation 
$\cK\otimes \one_{M_n(\C)}$ from 
$\cB(\fH)\otimes M_n(\C)$ to $\cB(\fH)\otimes M_n(\C)$:
$$\cK\otimes \one_{M_n(\C)} \left(\sum_{i,j=1}^n A_{i,j} \otimes E_{i,j}\right) = 
\sum_{i,j=1}^n \cK(A_{i,j}) \otimes E_{i,j}\ .$$
The map $\cK$ is completely positive in case for each $n\in \N$, 
$\cK\otimes \one_{M_n(\C)}$ is positivity preserving. The notion of complete positivity was introduced by Stinespring \cite{Sti55}.  Its physical relevance was discussed by Krauss \cite{K71}. See Paulsen's book \cite{Paul02} for more information on the mathematical theory. 

As an  example, it is very easy to see that for any $V\in \cB(\fH)$, the map $\cK$ defined by 
$\cK A = VAV^*$ is completely positive. It is also clear that any convex combination of 
completely positive maps is completely positive. 
Complete positivity arises naturally in quantum mechanics. A  {\em quantum Markov operator}
on $\cB(\fH)$ is a linear transformation $\cK$ on $\cB(\fH)$ that is completely positive and such that 
$\cK \one_{\fH} = \one_{\fH}$.   

If furthermore
\begin{equation}\label{symm}
\tr[(\cK A)^* B] = \tr[A^* \cK(B)]
\end{equation}
for all $A,B\in \cT_2(\fH)$, then $\cK$ is a {\em symmetric quantum Markov operator}.  
The Kadison \cite{Choi,Kad52} inequality says that for any completely positve operator $\cK$, and all $A\in \cB(\fH)$,
\begin{equation}\label{Kad}
\cK(A^*A) \geq (\cK(A))^*\cK(A)\ .
\end{equation}
(In fact, one only needs the positivity of $\cK \otimes \one_{M_n(\C)}$ for $n=2$.) In particular, when $\cK$ is a symmetic quantum Markov operator,
$$\tr[A^*A] = \tr[\cK(\one_\fH) A^*A]\ = \tr[\one_{\fH} \cK(A^*A)] \geq \tr[ (\cK(A))^*\cK(A)] \ .$$
That is, $\|\cK(A)\|_2 \leq \|A\|_2$, so that $\cK$ is a contraction on $\cT_2(\fH)$.
Let $\cK^\dagger$ denote the adjoint of $\cK$ with respect to the inner product on $\cT_2(\fH)$.

If $\cK$ is  a quantum Markov operator, then 
$\cK^\dagger$ is completely positive, and $\cK$ is called {\em normal} in case $\cK^\dagger$
maps $\Dens(\fH)$ into itself.  (Recall that  $\cT_1(\fH)$ is the predual of $\cB(\fH)$ in that every bounded linear functional on 
$\cT_1(\fH)$ is of the form $A\mapsto \tr[AB]$ for some $B\in \cB(\fH)$.)
Normal quantum Markov maps are often called {\em quantum operations} \cite{K71}.  
A {\em quantum Markov semigroup}  on $\cB(\fH)$  is a semigroup $\cP_t$  
of linear transformations on $\cB(\fH)$ such that  for all $A\in \cB(\fH)$ and all $\rho\in \Dens(\fH)$,
$\lim_{t\to 0}\tr[\rho \cP_tA] = \tr[\rho A]$, and such for each $t > 0$, 
$\cP_t$ is a normal quantum operator.   Of course when $\fH$ is finite dimensional, all 
vector space topologies on $\cB(\fH)$ are the same, and normality is automatic.

\section{The quantum Kac model}

\subsection{Binary collisions}

Consider a quantum mechanical system of $N$ particles. Let $\H$ be the state space of the single 
particle system. Let $\H_N$ denote the $N$-fold tensor product
$\H_N = \H^{\otimes N}$.    We write $\H_j$ to denote the $j$th factor of $\H$ in $\H_N$. If $A$ is an operator on 
$\H$, we write $A_{1}$ to denote the operator on $\H_N$ given by
\begin{equation}\label{tens1}
A_1 (\phi_1\otimes \phi_2\otimes \cdots \otimes \phi_N) = (A\phi_1) \otimes \phi_2 \otimes \cdots \otimes \phi_N\ .
\end{equation}
Likewise, for $j=2,\dots,N$, we define $A_j$ to act by applying $A$ to the $j$th factor only. 
Equivalently, let $\pi_{i}$ be the permutation of $\{1,\dots,N\}$  such that $\pi_{i}(i) =1$, and 
$\pi_{i}(k) =k$ for $k\neq 1,i$. Let $V_{i}$ be the canonical unitary representation of this permutation on $\H_N$. 
Then we have 
\begin{equation}\label{tens2}
A_{i} = V_{i} A_1 V_{i}^*\ .
\end{equation}

Let $h$ denote the single-particle Hamiltonian on $\H$.   We define the $N$-particle Hamiltonian $H_N$ by
\begin{equation}\label{ham1}
H_N = \sum_{j=1}^N h_j\ .
\end{equation}
No interactions between the particles are included in $H_N$ because it specifies the energy of a state {\em between} collisions, and apart from collisions, the particles do not interact. Thus, $H_N$ plays the same role as the
kinetic energy, which is the Hamiltonian for the free motion in the classical Kac model. As in the classical model, we will consider an evolution
in which the system makes random jumps from one state to another through binary collisions that preserve the energy of the non-interacting particles. All of the effects of the interaction are encoded into the description of these binary collisions as follows:

The {\em collision parameter space}  $\aC$ is a compact metric space; the elements of $\aC$ parameterize the kinematically possible collisions. In the classical Kac model discussed above $\aC = S^1$, the unit circle. We take as given a continuous map $U$ from $\aC$ to $\cU(\H\otimes \H)$, the group of unitary operators on $\H\otimes \H$.  Let $\sigma$ denote a generic point of $\aC$. Then $U(\sigma)$ may be regarded as the scattering matrix of a particular type of collision that two particles may undergo.  Let $\varrho$ be a density matrix on $\H\otimes \H$ representing the (non-interacting) state of the two particles before this collision. Then $U(\sigma)\varrho U^*(\sigma)$ gives the state  after the collision. 

To describe a binary collision between particles $1$ and $2$ in our $N$ particle system, define $U_{1,2}(\sigma)$
to be the unitary operator on $\H_N$ given by
$$U_{1,2}(\sigma) (\phi_1\otimes \phi_2\otimes\phi_3\otimes  \cdots \otimes \phi_N) = (U_{1,2}(\sigma)\phi_1 \otimes \phi_2)\otimes \phi_3 \otimes \cdots \otimes \phi_N\ .$$
In the same way, for $1\leq i < j \leq N$, we define $U_{i,j}(\sigma)$ so that it acts on the $i$th and $j$th factors of $\H$. Equivalently, let $\pi_{i,j}$ be the permutation of $\{1,\dots,N\}$  such that $\pi_{i,j}(i) =1$, $\pi_{i,j}(j) =2$, and 
$\pi_{i,j}(k) =k$ for $k\neq i,j,1,2$. Let $V_{i,j}$ be the canonical unitary representation of this permutation on $\H_N$. 
Then we have $U_{i,j}(\sigma) = V^*_{i,j} U_{1,2}(\sigma) V_{i,j}$.

Since we are only concerned with collisions that conserve energy, we require that for each $\sigma\in \aC$,
\begin{equation}\label{quankac1}
U(\sigma)[\one_\H \otimes h + h \otimes \one_\H] U^*(\sigma) = [\one_\H \otimes h + h \otimes \one_\H]\ .
\end{equation}
It then follows from the definitions that for each $\sigma\in \aC$, and each $1\leq i<j \leq N$
\begin{equation}\label{quankac2}
U_{i,j}(\sigma)H_NU_{i,j}^*(\sigma) = H_N\ ,
\end{equation}
or, what is the same, that each $U_{i,j}(\sigma)$ commutes with $H_N$.

\begin{defi}[Collision specification]\label{cfdef} A {\em collision specification} $(\aC, U,\nu)$ consists of a compact metric space $\aC$,  a continuous one-to-one function $U$ from $\aC$ to $\cU(\H_2)$, and a Borel probability measure $\nu$  that charges all open subsets of $\aC$  such that:

\smallskip
\noindent{\it (i)}  For each $\sigma\in \aC$, $U(\sigma)$ commutes with $H_2$.

\smallskip
\noindent{\it (ii)}   For some $\sigma_0\in \aC$, $U(\sigma_0) = \one_{\H_N}$. 

\smallskip
\noindent{\it (iii)} $\{U(\sigma) \ :\ \sigma\in \aC\} = \{U^*(\sigma) \ :\ \sigma\in \aC\}$ and the map $\sigma \mapsto
\sigma'$ where $U^*(\sigma) = U(\sigma')$ is a 
measurable transformation of $\aC$ that leaves $\nu$ invariant.

\smallskip
\noindent{\it (iv)} Let $V: \H_2\to \H_2$ be the swap transformation: $V \phi\otimes \psi = \psi\otimes \phi$ for all $\phi,\psi\in \H$.
Then  $\{U(\sigma) \:\ \sigma\in \aC\} = \{VU(\sigma)V^* \:\ \sigma\in \aC\}$ and the map $\sigma \mapsto
\sigma'$ where $VU(\sigma)V^* = U(\sigma')$ is a 
measurable transformation of $\aC$ that leaves $\nu$ invariant.   

\end{defi}

\begin{remark}  Since $U(\sigma)^*$ is the time reversal of the collision $U(\sigma)$, 
property {\it (iii)} in Definition~\ref{cfdef}  will  incorporate time-reversibility into the quantum Kac model. Property
{\it (ii)} together with the continuity of $U$ will mean that not only is the trivial collision included in the model, but also all ``grazing'' collisions.  The condition {\it (iv)} ensures that the two factors of $\H$ enter the collision specification in a symmetric way. 
\end{remark}

To proceed, we need to make some assumptions on the spectrum of the single-particle Hamiltonian $h$. We shall
always assume  that $\H$ is separable, and that $h$ has  a compact resolvent. Of course the latter condition is trivally satisfied when $\H$ is finite dimensional. The compactness of the resolvent  implies not only that $h$ has discrete spectrum, but that each of the eigenspaces of $h$ is finite dimensional. This is crucial in what follows, and it is much less restrictive than supposing that $\H$ be finite dimensional. For example, one could take $\H = L^2(\Omega)$ for a  regular bounded  domain $\Omega\subset \R^d$, and take $h = -\Delta$ with, say, Dirichlet conditions.  
We shall always write
$\{e_j\}_{j\in \cJ}$ to denote the sequence of eigenvalues of $h$ arranged in increasing order and repeated according to their multiplicity.  The index set $\cJ$ will be taken to be $\N$ when $\H$ is infinite dimensional, and otherwise it will be given by $\cJ =\{1,\dots, {\rm dim}(\H)\}$. We shall always write  $\{\psi_j\}_{j\in \cJ}$ denote an orthonormal basis of $\H$ consisting of eigenvectors of $h$ such that $h\psi_j = e_j\psi_j$ for each $j\in \cJ$.

The diagonalization of $h$ leads directly to a diagonalization of $H_N$. Let $\bma$ denote a generic element of $\cJ^N$ so that $\bma = (\alpha_1,\dots, \alpha_N)$. For each $\alpha \in \cJ^N$, define
\begin{equation}\label{quankac3}
E_\bma = \sum_{j=1}^N e_{\alpha_j} \quad{\rm and}\quad \Psi_\bma = \psi_{\alpha_1}\otimes \cdots \otimes \psi_{\alpha_N}\ .
\end{equation}
Evidently $\{\Psi_\bma\}_{\bma\in \cJ^N}$ is an orthonormal basis of $\H_N$, and for all $\bma$,
$H_N \Psi_\bma = E_\bma \Psi_\bma$. 

Note that even if each $h$ has non-degenerate spectrum, it will never be the case that 
$H_N$ has non-degenerate spectrum.  

\begin{defi} [Energy shells]
Let $\spec$ denote the spectrum of $H_N$. For each $E\in H_N$, we call ``energy shell at energy $E$'' the eigenspace  $\K_E$ of $H_N$
with eigenvalue $E$. For each $E\in \spec$, let $P_E$ denote the orthogonal projection in $\H_N$ 
onto $\K_E$ so that  
\begin{equation}\label{quankac4}
\H_N = \bigoplus_{E \in \spec} \K_E\ .
\end{equation}
\end{defi}

On account of the compactness assumption on the resolvent of $h$,  each $\K_E$ is finite dimensional. 
The different energy shells $\K_E$ correspond to the level surfaces of the classical energy function $E(\vec v)$
defined in (\ref{clkac1}), which are the energy spheres $\SS_E$. Evidently each 
$\K_E$ is invariant under each $U_{i,j}(\sigma)$, or, what is the same,
\begin{equation}\label{quankac5}
U_{i,j}(\sigma) P_E U_{i,j}^*(\sigma)  = P_E\ .
\end{equation}
For each $E\in \spec$, define
\begin{equation}\label{shell}
\sigma_E = \frac{1}{{\rm dim}(\K_E)}P_E\ .
\end{equation}
Then $\sigma_E \in \Dens(\H_N)$, and it is the analog of the uniform measure on the classical energy shell
$\SS_E$.  Of course it is only possible to define $\sigma_E$ because $\K_E$ is finite dimensional.

An important feature of the classical Kac model is that for each $E>0$, the only continuous functions $F$
on $\SS_E$ that satisfy $F = F\circ R_{i,j,\theta}$ for all $1\leq i< j \leq N$ and $\theta\in [-\pi,\pi]$ are the constants. 
Equivalently, a  continuous function $F$ on $\R^N$ satisfies $F = F\circ R_{i,j,\theta}$ for all 
$1\leq i< j \leq N$ and $\theta\in [-\pi,\pi]$ if and only if $F$ depends on $\vec v\in \R^N$ only through 
$E(\vec v)$. This feature of the classical collisions provides the ergodicity of the Kac walk on the energy spheres.  
We require its quantum analog.  

\begin{defi}[Energy algebra]  Let $\cA_{N}$ denote the commutative subalgebra of 
$\cB(\H_N)$ generated by $\{P_E\ :\ E\in \spec\}$. The elements of $\cA_N$ are the elements of the form 
$f(H_N)$ for some bounded continuous function
$f: \spec \to \C$, or what is the same thing,
\begin{equation}\label{expand}
\sum_{E\in \spec}\lambda_E P_E
\end{equation}
where $\lambda_E = f(E)$.   $\cA_N$ is called the {\em energy algebra}. 
\end{defi} 

The elements of $\cA_N$ are the quantum analogs of the functions $F$  on   $\R^N$ that depend on 
$\vec v\in \R^N$ only through the energy $E(\vec v)$.  Note that if $A\in \cA_N$, then 
$A$ has an expansion of the from (\ref{expand}) where
\begin{equation}\label{expand2}
\lambda_E  = \tr[\sigma_E A] 
\end{equation}
and $\sigma_E$ is given by (\ref{shell}).

For any ${\mathcal S} \subset \cB(\K)$, $\K$ any Hilbert space,  ${\mathcal S}'$ denotes the commutant of ${\mathcal S}$ in 
$\cB(\K)$.
By part {\it (i)} of Definition~\ref{cfdef},
\begin{equation}\label{comcon}
\cA_2 \subset   \{U(\sigma)\ :\ \sigma\in \aC\}'\ .
\end{equation}
 We are primarily interested in collision specifications such that 
$\{U(\sigma)\ :\ \sigma\in \aC\}'=\cA_2$.

\begin{defi} [Ergodic collision specification]\label{ergdef} A collision specification $(\aC,U,\nu)$ as in Definition~\ref{cfdef} is {\em ergodic}
 in case 
\begin{equation}\label{quankac6}
\{U(\sigma)\ :\ \sigma\in \aC\}'=\cA_2\ .
\end{equation}
\end{defi}
Alternatively we can say that a collision specification is ergodic if and only if whenever $A$  commutes with each $U(\sigma)$, then $A$ is a function
of $H_2$.  Let $A\in \cB(\H_2)$. Then since $H_2$ has discrete spectrum, $A\in \cA_2$ if and only if $P_E A\in \cA_2$ for each eigenvalue $E$ of $H_2$. 
Since the eigenspaces of $H_2$ are all finite dimensional,  for all  $A\in \cB(\H_2)$ and all eigenvalues $E$ of $H_2$, $P_E A \in \cT_2(\H_2)$.   It follows that
\eqref{quankac6} is satisfied if and only if
\begin{equation}\label{quankac6B}
\cT_2(\H_2)\cap \{U(\sigma)\ :\ \sigma\in \aC\}'=\cT_2(\H_2)\cap \cA_2\ .
\end{equation}
Thus, \eqref{quankac6B} gives another characterization of ergodicity, and it has the advantage that since $\cT_2(\H_2)$ is a Hilbert space, it can be related to an eigenvalue problem. 

\begin{defi}[Collision operator]  Let $(\aC,U,\nu)$ be a collision specification. Define the {\em collision operator} $\cQ$ on $\cB(\H_2)$ by 
\begin{equation}\label{qdef}
\cQ A = \int_{\aC} {\rm d}\nu(\sigma) U(\sigma)AU^*(\sigma)\ .
\end{equation}
\end{defi}

It is evident that the restriction of $\cQ$ to $\cT_2(\H_2)$ is self-adjoint, and that for all $A\in \{U(\sigma)\ :\ \sigma\in \aC\}'$, $\cQ(A) =A$.
Therefore, $\cT_2(\H_2)\cap \{U(\sigma)\ :\ \sigma\in \aC\}' $ is contained in the eigenpace of $\cQ$ with eigenvalue $1$.
The next lemma says, in particular, that $\cT_2(\H_2)\cap \{U(\sigma)\ :\ \sigma\in \aC\}' $ {\em is precisely}  the eigenpace of $\cQ$ with eigenvalue $1$.
The lemma shall have other uses as well, and we state is in a more general form that we need at present to avoid repetition later on.

\begin{lm}[Convexity lemma]\label{convlem}  Let $(\aC,U,\nu)$ be a collision specification.Let $\Phi$ be a convex function on $\cB(\H_2)$ with the property
that for all $U\in \cU(\H_2)$ and all $A\in \cB(\H_2)$, $\Phi(U A U^*) = \Phi(A)$. Then 
\begin{equation}\label{convlem2}
\Phi(\cQ A) \leq \Phi(A)
\end{equation}
and if $\Phi$ is strictly convex,  there is equality in (\ref{convlem2}) with $\Phi(A)< \infty$  if and only if 
$A\in \{U(\sigma)\ :\ \sigma\in \aC\}'$.   In particular, taking $\Phi(A) = \tr[A^*A]$,  the eigenpace of $\cQ$ with eigenvalue $1$ is $\cT_2(\H_2)\cap \{U(\sigma)\ :\ \sigma\in \aC\}' $. 
\end{lm}

\begin{proof} By the convexity and unitary invariance, 
\begin{equation}\label{ave}
\Phi(\cQ A)  \leq  \int_{\aC} {\rm d}\nu(\sigma)
\Phi(U(\sigma) A U(\sigma)^*)  =    \int_{\aC} {\rm d}\nu(\sigma)
\Phi( A )  = \Phi(A)\ ,
\end{equation}
which proves (\ref{convlem2}). If there is equality in (\ref{convlem2}) and $\Phi$ is strictly convex, we must have
that  $U(\sigma) A U^*(\sigma)$ is constant almost everywhere with respect to $\nu$, and then by
the continuity of $\sigma \mapsto U(\sigma)$ and  the fact that $\nu$ charges all open sets,   $U(\sigma) A U^*(\sigma)$ is independent of $\sigma$. Since for some $\sigma_0$, $U(\sigma_0)  = \one_{\H_2}$,  it must be the case that 
$U(\sigma) A U^*(\sigma) = A$ for all $\sigma\in \aC$. The final statement then follows from the remarks made above. 
\end{proof}

It follows from Lemma~\ref{convlem} that for $A\in \cT_2(\H_2)$, $\lim_{n\to\infty}\cQ^n A = PA$ where $P$ is the orthogonal projection onto
$\{U(\sigma)\ :\ \sigma\in \aC\}'$.  Thus, to prove ergodicity, it suffices to show that $\lim_{n\to\infty}\cQ^n A \in \cA_2$ for all $A\in \cT_2(\H)$.

 In the next two examples, it is possible to derive an explict form for $\cQ$ from which one can easily determine the 
 eigenspace with eigenvalue $1$, and thus verify ergodicity.  

\begin{exam}\label{exam2} For the simplest possible example, take $\H = \C^2$,  so that $\H_N = (\C^2)^{\otimes N}$.   Define the single particle Hamiltonian $h$ by 
  $h = \left[\begin{array}{cc} 0 & 0\\ 0 & 1\end{array}\right]$ so that the $N$-particle Hamiltonian
 $H_N = \sum_{j=1}^N h_j$ has $\spec   =\{0,\dots,N\}$.  For  $E \in \{0,\dots,N\}$, 
 $${\rm dim}(\K_E) = \tbinom{N}{E}\ .$$

Identify $\C^2\otimes \C^2$ with $\C^4$ using the basis
$$\left(\begin{array}{c} 1\\ 0\end{array}\right) \otimes  \left(\begin{array}{c} 1\\ 0\end{array}\right) \ ,\quad
\left(\begin{array}{c} 0\\ 1\end{array}\right) \otimes  \left(\begin{array}{c} 1\\ 0\end{array}\right) \ ,\quad
\left(\begin{array}{c} 1\\ 0\end{array}\right) \otimes  \left(\begin{array}{c} 0\\ 1\end{array}\right) \ ,\quad
\left(\begin{array}{c} 0\\ 1\end{array}\right) \otimes  \left(\begin{array}{c} 0\\ 1\end{array}\right) \ .$$
The standard physics notation for this basis is simply
\begin{equation}\label{2basis}
|00\rangle\ , \quad   |10\rangle\ ,  \quad  |01\rangle\ , \quad   |11\rangle\ ,
\end{equation}
which will be useful.  With this identification of $\C^2\otimes \C^2$ with $\C^4$,
$$
\left[\begin{array}{cc} a_{1,1} & a_{1,2}\\ a_{2,1} & a_{2,2}\end{array}\right] \otimes
\left[\begin{array}{cc} b_{1,1} & b_{1,2}\\ b_{2,1} & b_{2,2}\end{array}\right] =: A \otimes B  \quad{\rm is \ represented\ by}\quad 
\left[\begin{array}{cc} b_{1,1}A & b_{1,2}A\\ b_{2,1}A & b_{2,2}A\end{array}\right] \ .
$$
(Switching the order of the second and third basis elements swaps the roles of $A$ and $B$ in the block matrix representation of the tensor product $A\otimes B$.)
 In the sort of notation used in (\ref{2basis}), an orthonormal  basis of $\H_N$ consisting of eigenvectors of $H_N$ 
is provided by the set of vectors $|\alpha_1,\dots,\alpha_N\rangle$
 in which each $\alpha_j$ is either $0$ or $1$. Then
 $$H_N  |\alpha_1,\dots,\alpha_N\rangle  =\left(\sum_{j=1}^N \alpha_j\right) |\alpha_1,\dots,\alpha_N\rangle\ .$$

In this basis,
$$H_2=\left[\begin{array}{cc} 0 & 0\\ 0 & 1\end{array}\right] \otimes I + I\otimes \left[\begin{array}{cc} 0 & 0\\ 0 & 1\end{array}\right] = 
\left[\begin{array}{cccc} 0 & 0 & 0 & 0\\
0 & 1 & 0 & 0\\
0 & 0 & 1 & 0\\
0 & 0 & 0 & 2\end{array}\right]\ .$$
Therefore,
${\rm Spec}(\H_2) = \{0,1,2\}$ and
$$
P_0 = \left[\begin{array}{cccc} 
1 & 0 & 0 & 0\\
0 & 0 & 0 & 0\\
0 & 0 & 0 & 0\\
0 & 0 & 0 & 0\end{array}\right]\ , \quad 
P_1 = \left[\begin{array}{cccc} 
0 & 0 & 0 & 0\\
0 & 1 & 0 & 0\\
0 & 0 & 1 & 0\\
0 & 0 & 0 & 0\end{array}\right]\ \quad {\rm and}\quad
P_2 = \left[\begin{array}{cccc} 
0 & 0 & 0 & 0\\
0 & 0 & 0 & 0\\
0 & 0 & 0 & 0\\
0 & 0 & 0 & 1\end{array}\right]\ .$$

Now define $\aC = S^1\times S^1\times S^1 \times S^1 $ identifying each copy of $S^1$ with the unit circle in $\C$
so that the general point in $\sigma \in \aC$ has the form \ $\sigma = (e^{i\varphi},e^{i\theta},e^{i\psi} , e^{i\eta})$. Then define
$$U(e^{i\phi})  :=   \left[\begin{array}{cccc} e^{i\theta } & 0 &  \phantom{-}0 & 0\\
0 & e^{i\psi}\cos\theta  & -e^{i\varphi}\sin\theta  & 0\\
0 & e^{-i\varphi}\sin\theta  & \phantom{-}e^{-i\psi}\cos\theta & 0\\
0 & 0 &  \phantom{-}0 & e^{i\eta} \end{array}\right]\ $$
Choosing $\nu$ to be the uniform probability measure on $\aC$ gives us a collision specification $(\aC,U, \nu)$.

A simple computation shows that for every operator $A$ on $\H_2 = \C^2\otimes \C^2$ identified as the $4\times 4$ matrix with entries $a_{i,j}$ using the basis (\ref{2basis}),
\begin{eqnarray}\label{Qform}
\cQ A = \int_{\aC} {\rm d}\nu(\sigma) U(\sigma)AU^*(\sigma) &=&
 \left[\begin{array} {cccc} 
a_{1,1} & 0 & 0 & 0\\
0 & \frac12(a_{2,2}+a_{3,3}) & 0 & 0\\
0 &  0 &  \frac12(a_{2,2}+a_{3,3}) & 0\\
0 & 0 & 0 &  a_{4,4}\end{array}\right]\nonumber\\  
&= & a_{1,1} P_0 + \frac{a_{2,2}+ a_{3,3}}{2} P_1  + a_{4,4} P_2 \in \cA_2\ .
\end{eqnarray}
Therefore,
$$\{U(\sigma)\ :\ \sigma\in \aC\}' \subset {\rm ran}(\cQ) \subset \cA_2 \subset \{U(\sigma)\ :\ \sigma\in \aC\}'\ ,$$
showing that $(\aC,U,\nu)$ is ergodic.  
\end{exam} 

\begin{exam}\label{exam2A} For the next simplest example, we take $\aC$ and $U$ as in the previous example, but we take 
$\nu$ to be a non-uniform probability measure on $\aC$. 
For example, take
$$\nu = (2\pi)^{-4}(1+\cos\varphi)(1+\cos\theta)(1+\cos\psi)(1+\cos\eta){\rm  d}\varphi{\rm  d}\theta {\rm  d}\psi {\rm  d}\eta\ .$$
It is easy to check that conditions {\it (i)} through {\it (iv)} are satisfied. 
Then (\ref{Qform}) becomes
\begin{equation}\label{notproj}
\cQ A  =
 \left[\begin{array} {cccc} 
a_{1,1} & \frac18 a_{1,2} & \frac18 a_{1,3} & \frac12 a_{1,4}\\
\frac18 a_{2,1} & \frac12 (a_{2,2}+a_{3,3}) & 0 & \frac14 a_{2,4}\\
\frac18 a_{3,1}  &  0  &  \frac12(a_{2,2}+a_{3,3}) & \frac14 a_{3,4}\\
\frac12 a_{4,1} & \frac14 a_{4,2} & \frac14 a_{3,4} &  a_{4,4}\end{array}\right]\ .
\end{equation}
In this case, $\cQ A\notin \cA_2$. However, it is clear that $\lim_{n\to\infty}\cQ^nA$ is, 
and hence $(\aC,U,\nu)$ is ergodic.

\end{exam}

\subsection{The quantum Kac generator}

We are now ready to define the quantum analogs of classical transformations ${\displaystyle F \mapsto  \frac{1}{2\pi}\int_{-\pi}^\pi F(R_{i,j,\theta}\vec v)\dd \theta}$
defined in (\ref{clkac3}).  Given a collision specification $(\aC,U,\nu)$, for each $N\geq 2$ define a family of operators on $\cB(\H_N)$, $\{\cQ_{i,j}\}_{1\leq i,j\leq N}$ as follows:  For all $A\in \cB(\H_N)$, 
\begin{equation}\label{qijdef}
\cQ_{i,j}A =  \int_{\aC} {\rm d}\nu(\sigma)
U_{i,j}(\sigma) A U_{i,j}^*(\sigma)\ .
\end{equation}
Evidently, for each $i,j$, $\cQ_{i,j} \one_{\H_N} =  \one_{\H_N} $ and $\cQ_{i,j}$ preserves positivity. If $A$ is trace-class, then so is $\cQ_{i,j}A$ and $\tr[cQ_{i,j}A] = \tr[A]$.
That is, each $\cQ_{i,j}$ is both a quantum Markov operator, and, when restricted to $\Dens(\H_N)$, a quantum operation, as defined in the final paragraph of Section 1.2.  Moreover, because of the way that time reversibility has been incorporated into the definition of
$(\aC,U,\nu)$, each $\cQ_{i,j}$ is a {\em self-adjoint} quantum Markov operator. 
In particular, each 
$\cQ_{i,j}$ is a self-adjoint contraction on $\cT_2(\H_N)$.

The transformation $A \mapsto \cQ_{i,j}\varrho$ is a quantum analog of the  
classical transformation defined in (\ref{clkac3}), and this brings us to the definition of the quantum Kac generator:
\begin{defi} [Quantum Kac generator]\label{gendef} Let $\{U(\sigma)\:\ \sigma\in \aC\}$ be an ergodic set of collision operators
and let $\nu$ be a given Borel probability measure on $\aC$. Define the operators $\cQ_N$ and $\cL_N$ on $\cB(\H_N)$ by
\begin{equation}\label{quankac9}
\cQ_N  =  {\ncht}^{-1}\sum_{i<j} \cQ_{i,j} \quad{\rm and}\quad 
\cL_N  = N(\cQ_N - \one_{\H_N})\ .
\end{equation} 
\end{defi}

Note that by property {\it (iv)} in Definition~\ref{cfdef},  for all $i\neq j$, 
\begin{equation}\label{symrem}
\cQ_{i,j} = \cQ_{j,i}\ .
\end{equation}
Hence one has the alternate formula for $\cQ_N$:
\begin{equation}\label{quankac9B}
\cQ_N  =  \frac{1}{N(N-1)} \sum_{i\neq j} \cQ_{i,j} \ .
\end{equation} 

$\cQ_N$ is a self-adjoint quantum Markov operator because the set of such operators is convex, and $\cQ_N$ is a convex  combination of the $\cQ_{i,j}$ which are symmetric quantum Markov operators.  In particular,
$\cQ_N$ is a contraction on $\cT_2(\H_N)$. For the same reason,  the restriction of  $\cQ_N$
to $\Dens(\H_N)$ is a quantum operation.

The {\em Quantum Kac Master Equation} (QKME) is the evolution equation on $\Dens(\H_N)$ given by
\begin{equation}\label{quankac12}
\frac{{\rm d}}{{\rm d} t} \varrho(t)  = \cL_N \varrho(t) \ .
\end{equation}

Since $\|\cL_N\|_\infty \leq 2N$, the QKME is solved by exponentiation:
For each $t\geq 0$ ,we may define an operator $\cP_{N,t}$ on each  $\cT_p(\H_N)$, and in particular on $\Dens(\H_N)$, by 
\begin{equation}\label{quankac11}
\cP_{N,t} A  = \sum_{k=1}^\infty e^{-Nt}\frac{(Nt)^k}{k!} \cQ^k_N A  = e^{t\cL_N} A\ .
\end{equation} 
Then the unique solution $\varrho(t)$ of the QKME satisfying $\varrho(0) = \varrho_0\in \Dens(\H_N)$ is
$\varrho(t) = \cP_{N,t}\varrho_0$.

The first equality in (\ref{quankac11}) displays each   $\cP_{N,t}$ as a convex combination of the $\cQ_N$. Thus,
as above, each $\cP_{N,t}$ is both a symmetric quantum Markov operator (and hence a contraction on 
$\cT_2(\H_N)$) and a quantum operation. 
Hence $\cP_{N,t}$ is a  quantum Markov semigroup.  

\subsection{Permutation invariance}

The symmetric group ${\mathcal S}_N$ of permutations $\pi$ of $\{1,\dots,N\}$ has a natural unitary action on $\H_N$ that commutes
with $\cL_N$.   Let $\phi_1\otimes \cdots \otimes \phi_N$ be a product vector in $\H_N$. For $\pi\in {\mathcal S}_N$,
define 
\begin{equation}\label{Upidef}
U_\pi(\phi_1\otimes \cdots \otimes \phi_N) = \phi_{\pi(1)}\otimes \cdots \otimes \phi_{\pi(N)}\ ,
\end{equation}
and then extend $U_\pi$ by linearity to produce a linear operator on $\H_N$. Evidently $U_\pi$ is unitary.  

We will be especially interested in density matrices $\varrho$ such that 
$$U_\pi \varrho U_\pi^* = \varrho$$
for all $\pi\in {\mathcal S}_N$. We call such density matrices {\em symmetric}. There is an orthonormal  basis for $\cT_2(\cH_N)$ consisting of operators of the form 
$A_1\otimes \cdots \otimes A_N$ with each $A_j$, $j=1,\dots,N$ chosen from an orthonormal basis for 
 $\cT_2(\H)$. It follows directly from \eqref{Upidef} that 
\begin{equation}\label{Upidef2}
U_{\pi} (A_1\otimes \cdots \otimes A_N) U^*_\pi = A_{\pi(1)}\otimes \cdots \otimes A_{\pi(N)}\ .
\end{equation}
Then from \eqref{Upidef2}, 
It is evident that for all $A\in \cB(\H_N)$, and all $1\leq i < j \leq N$, 
\begin{equation}\label{Upidef3}
 \cQ_{i,j} U_{\pi} A U^*_\pi   = \cQ_{\pi(i),\pi(j)}A\ ,
 \end{equation}
and then from \eqref{quankac9B}, it follows that $\cQ_N U_{\pi} A U^*_\pi = \cQ_N A$. In particular, if $\varrho$ is symmetric, then
$\cP_{N,t}\varrho$ is symmetric for each $t>0$. Furthermore. by the self-adjointness of the operators  $\cQ_{i,j}$, for any 
$B\in \cB(\H_N)$, 
$$\tr[B^*( \cQ_{i,j} U_{\pi} A U^*_\pi )] = \tr[B^*\cQ_{\pi(i),\pi(j)}A] = \tr[(\cQ_{\pi(i),\pi(j)}B)^*A]$$
and
$$\tr[B^*( \cQ_{i,j} U_{\pi} A U^*_\pi )] = \tr[(\cQ_{i,j}B)^* U_{\pi} A U^*_\pi )] = \tr[(U_\pi^* \cQ_{i,j}BU_\pi)^* A  )] \ .$$
Therefore,
\begin{equation}\label{Upidef4}
U_{\pi}^* (\cQ_{i,j}  A) U_\pi   = \cQ_{\pi(i),\pi(j)}A\ .
 \end{equation}
 Notice the different orders of $U_\pi$ and $U^*_\pi$ in \eqref{Upidef3} and \eqref{Upidef4}, which can also be understood in terms of a replacement of $\pi$ with $\pi^{-1}$.

\subsection{The nullspace of the the quantum Kac generator}

The analog of the Kac conjecture for the QKME concerns the long time behavior of its solutions. A basic first step in the investigation of the long time behavior is to determine all of the steady-state solutions. In this section we give a characterization of the null space of $\cL_N$.   Of course since $\cL_N = N(\cQ_N -\one)$, this is the same as the eigenslace of $\cQ$ with eigenvalue $1$. The following simple generalization of Lemma~\ref{convlem} will be useful. 
Note that when $N =2$, $\cQ_{1,2}$ is exactly the operator $\cQ$ of Lemma~\ref{convlem}. The adaptation to higher $N$ is trivial and will be left to the reader.

\begin{lm}[Convexity lemma on $\cB(\H_N)$]\label{convlem3}  Let $\Phi$ be a convex function on $\cB(\H_N)$ with the property
that for all $U\in \cU(\H_N)$  where $ \cU(\H_N)$ is the unitary group on $\H_N$, and all $A\in \cB(\H_N)$, $\Phi(U A U^*) = \Phi(A)$. Then
for all $1\leq i< j \leq N$, 
\begin{equation}\label{convlem3}
\Phi(\cQ_{i,j}A) \leq \Phi(A)
\end{equation}
and if $\Phi$ is strictly convex,  there is equality in (\ref{convlem3}) with $\Phi(A)< \infty$  if and only if $U(\sigma)A U^*(\sigma) = A$ for all $\sigma\in \cA$. 
\end{lm}

\begin{remark} We do not assume that  $\Phi$ is finite everywhere on $\cB(\H_2)$. In our first application, we will take 
$\Phi(A) = \tr[A^*A]$ which is finite if and only if $A\in \cT_2(\H_N)$. In this case, $\Phi$ is finite on all of $\Dens(\H_N)$.
 Later we shall consider an example based on the von Neumann entropy:
\begin{equation}\label{vnen}
\Phi(A) = \begin{cases} \tr[A \log A] &  A \in \Dens(\H_N)\\
\infty & A\notin \Dens(\H_N)\ . \end{cases}
\end{equation}
When $\H$ is infinite dimensional, this function is infinite for certain $\varrho\in \Dens(\H_N)$.
\end{remark}

\begin{lm}[Spectrum of $\cL_N$]\label{specL} Let $(\aC,U,\nu)$ 
be  a collision specification, and let $\cL_N$ and $\cQ_N$
be defined in terms of it as in (\ref{quankac9}). 
$\cQ_N$ and  $\cL_N$ have discrete spectrum: There is complete orthonormal basis of  $\cT_2(\H_N)$
consisting of eigenvectors of $\cQ_N$ and $\cL_N$.  Moreover, ${\rm Spec}(\cQ_N) \subset(0,1]$,  
 and ${\rm Spec}(\cL_N) \subset (-N,0]$.  
The null space of $\cL_N$ in $\cT_2(\H_N)$, ${\rm Null}(\cL_N)$,  is given by
\begin{equation}\label{nullspec}
{\rm Null}(\cL_N) = \{ A\in \cT_2(\H_N)\ :\ U_{i,j}(\sigma)AU^*_{i,j} = A \quad{\rm all}\ 1\leq i < j \leq N,\  \sigma\in \aC \}\ .
\end{equation}
\end{lm}

\begin{proof} That $\cQ_N$ is positive follows from 
$$
{\rm Tr} (A^* \cQ A) =  \int_\aC d \nu {\rm Tr} (A^* UAU^*) =  \int_\aC d \nu {\rm Tr}((UA)^* (UA)) \ge 0
$$ 
with equality only if $A=0$. Since ${\rm Tr}(A^* \cQ A)^2 \le {\rm Tr}(A^* A){\rm Tr}((\cQ A)^* \cQ A)$ we know from Lemma~\ref{convlem}
that ${\rm Tr}(A^* \cQ A) \le {\rm Tr}(A^*A)$ and hence ${\rm Spec}(\cQ_N) \subset (0,1]$.  This readily implies that
${\rm Spec}\cL_N \subset (-N,0]$.
For each $E,E'\in \spec$, let $\cX_{E',E}$ denote the subspace of operators $X$ on $\H_N$ such that
the range of $X$ is contained in $\K_{E'}$ and the range of $X^*$ is contained in $\K_E$. The different $\cX_{E',E}$
are mutually orthogonal in $\cT_2(\H_N)$, and
$$\cT_2(\H_N) = \bigoplus_{E,E'\in \spec} \cX_{E',E}\ .$$
Since $\K_{E'}$ and $\K_E$ are finite dimensional, $\cX_{E',E}$ is finite dimensional. Since 
$\K_{E'}$ and $\K_E$ are invariant under each $U_{i,j}(\sigma)$, $\cX_{E',E}$ is invariant under $\cQ_{i,j}$ for each
$i$ and $j$, and hence also under $\cQ_N$. Since $\cQ_N$ is self-adjoint, 
it may be diagonalized on each of these finite dimensional subspaces.  The same applies to $\cL_N = N(\cQ_N -\one_{\H_N} )$.

Since each $A\in \cA_N$ is invariant under each $\cQ_{i,j}$, it is clear that $\cA_N \subset {\rm Null}(\cL_N)$. 
Now suppose that $(\aC,U,\nu)$ is ergodic. Let $\Phi$ denote the strictly convex function $\Phi(A) = \tr[A^*A]$ which is finite everywhere on $\cT_2(\H_N)$. 
By (\ref{quankac9}) which displays $\cQ_N$ as a convex combination of the $\cQ_{i,j}$ with strictly positive weights, for all $A\in \cB(\H_N)$, Lemma~\ref{convlem} implies that $\Phi(\cQ_N A) \leq \Phi(A)$
and there is equality if and only if $\Phi(\cQ_{i,j} A) = \Phi(A)$ for each $i,j$.  Since 
$\Phi$ is strictly convex and finite, Lemma~\ref{convlem} further implies that if $\Phi(\cQ_N A) = \Phi(A)$, then
$U_{i,j}(\sigma) A U_{i,j}^*(\sigma) = A$ for all $\sigma$ and all $i,j$. 
Hence if $\cQ_N A =  A$, then $U_{i,j}(\sigma) A U_{i,j}^*(\sigma) = A$ for all $\sigma$ and all $i,j$ which proves \ref{nullspec}. 
 \end{proof}

\begin{defi}  For each $N$, let $\cC_N$ be the commutant
\begin{equation}\label{commutant}
\cC_N = \{ U_{i,j}(\sigma)\ : 1\leq i < j \leq N, \sigma\in \aC\}'\ .
\end{equation} 
By (\ref{nullspec}), $\cC_N \cap \cT_2(\H_N)$ is precisely the null space of $\cL_N$ acting on $\cT_2(\H_N)$. 
Let 
$\EC$ denote the orthognal projection in $\cT_2(\H_N)$ onto $\cC_N\cap \cT_2(\H_N)$. The notation recalls the fact that $\cC_N$ is an algebra as well as a subspace, and that this operation can be considered as a non-commutative conditional expectation. 
\end{defi}

Evidently, $\cA_N \subset \cC_N$ and $\cA_2 = \cC_2$.  However, it may be that $\cC_N$ is strictly larger than $\cA_N$. In this case, there are observables that are conserved by the evolution that are not functions $f(\H_N)$ of the energy alone. This may happen even if the collision specification $(\aC,U,\nu)$  is ergodic: Ergodicity at the $2$ particle level of individual collisions may or may not imply ergodicity at the $N$ particle level, as we shall see.

In any case, let $(\aC,U,\nu)$ be an ergodic collision specification, and let $\cL_N$ and $\cQ_N$
be defined in terms of it as in (\ref{quankac9}). By Lemma~\ref{specL}, $\EC$ is the orthognal projection onto ${\rm Null}(\cL_N)$, and moreover, for all $A\in \cT_2(\H_N)$, 
\begin{equation}\label{approach}
\lim_{t\to\infty} \cP_{N,t} A = \EC A\ .
\end{equation}
The steady states of the QKME are evidently the density matrices $\varrho\in \Dens(\H_N)$
that satisfy $\varrho = \EC \varrho$. 

The argument leading to (\ref{approach}) gives no information on the rate of convergence. In a later paper we investigate rates in terms of entropy production inequalities and spectral gaps.  In order to do this, it 
is first neccessary to obtain a more explicit description of $\cC_N$.
In particular, we would like to know when $\cC_N = \cA_N$, and, if this is not the case, how much larger $\cC_N$ may be than $\cA_N$.   The next lemma shows that while $\cC_N$, may be larger that $\cA_N$, at least, like $\cA_N$, it is commutative.

\begin{thm}[$\cC_N$ is commutative]\label{echeck}  Let $(\aC,U,\nu)$ be a $2$-ergodic collision specification.  Then $\cC_N$ is a commutative algebra.  In fact, every $A \in \cC_N$ is diagonal in the basis $\{\Psi_\bma\}$.
\end{thm}

\begin{proof}  Let $A$ be a self-adjoint operator in $\cC_N$. It suffices to show that  for each $\bma$, $\Psi_\bma$ is an eigenvector of $A$.

Recall that we have defined the operator $\cQ$ on $\cB(\H_2)$ by
\begin{equation}\label{echeck2}
\cQ A =  \int_{\aC} {\rm d}\nu(\sigma)U(\sigma) A U^*(\sigma)\ .
\end{equation}
We may think of $\cQ$ as being $\cQ_{1,2}$ for $N=2$  and then by Lemma~\ref{specL},
$\cQ$ is a symmetric quantum Markov operator, and hence a 
contraction on $\cT_2(\H_2)$ as well as $\cB(\H_2)$,
and the eigenspace of $\cQ$ with eigenvalue $1$ in $\cT_2(\H_2)$ is precisely the commutant 
$\{ U(\sigma)\ :\ \sigma\in \aC\}'$

Since $\cL A = 0$, $\cQ_{i,j}A = A$ for each $1 \leq i < j \leq N$. 
For $N \geq 3$, we may identify $\H_N$ with $\H_2\otimes \H_{N-2}$ where the first factor
corresponds to the first two factors of $\H$, and the second factor to the remaining factors of $\H$. Corresponding to this identification, we may write
$$\cQ_{1,2} = \cQ \otimes \one_{\H_{N-2}}\ .$$
This product structure gives rise to a simple description of ${\rm Null}(\cQ_{1,2} -\one_{\H_N})$  in $\cT_2(\H_N)$: 
Let $\{X_j\}$ be an orthonormal basis for the eigenspace of ${\rm Null}(\cQ -\one_{\H_2})$ in $\cT_2(\H_2)$, and let $\{Y_k\}$ be an orthonormal basis of $\cT_2(\cH_{N-2})$.  Then $\{ X_j\otimes Y_k\}$ is an orthonormal basis for 
${\rm Null}(\cQ_{1,2} -\one_{\H_N})$  in $\cT_2(\H_N)$. 

When $(\aC,U,\nu)$ is $2$-ergodic, the fact that $X_j\in {\rm Null}(\cQ -\one_{\H_2})$ means that for some function $f$,
$X_j= f(H_2)$. It follows that for all $\bma$,
\begin{eqnarray*}
X_j\otimes Y_k \Psi_\bma 
&=& X_j (\psi_{\alpha_1}\otimes \psi_{\alpha_2})\otimes Y_k( \psi_{\alpha_3}\otimes \cdots \otimes \psi_{\alpha_N})\nonumber\\
&=& f(H_2) (\psi_{\alpha_1}\otimes \psi_{\alpha_2})\otimes Y_k( \psi_{\alpha_3}\otimes \cdots \otimes \psi_{\alpha_N})\nonumber\\
&=& f(e_{\alpha_1} + e_{\alpha_2}) (\psi_{\alpha_1}\otimes \psi_{\alpha_2}) \otimes Y_k( \psi_{\alpha_3}\otimes \cdots \otimes \psi_{\alpha_N})\nonumber
\end{eqnarray*}
It follows that $\langle \Psi_\bmb, X_j\otimes Y_k \Psi_\bma\rangle_{\H_N} = 0$ unless $\beta_1=\alpha_1$ and $\beta_2 = \alpha_2$.  
Since this is true for each $j$ and $k$, then, whenever $A\in \cT_2(\H_2)$ and $\cQ_{12}A = A$, 
$\langle \Psi_\bmb, A \Psi_\bma\rangle_{H_N} = 0$ unless $\beta_1=\alpha_1$ and $\beta_2 = \alpha_2$.  Then by symmetry in the indices, if $A\in \cT_2(\H_N)$ satisfies $\cQ_{i,j}A= A$ for all $1\leq i < j \leq N$, 
$\langle \Psi_\bmb, A \Psi_\bma\rangle_{H_N} = 0$ unless $\bmb = \bma$. This proves that for all self-adjoint $A\in {\rm Null}(\cL_N)$,
$A \Psi_\bma = \lambda_\bma \Psi_\bma$ for some $\lambda_\bma\in \R$.
\end{proof}

It is well known that a commutative von Neumann algebra on a separable Hilbert space is generated by the spectral projection of a single self adjoint
operator. In the present case, we have an even more favorable situation:  We can give an explicit description of all of the minimal projections in 
$\cC_N$; this is done in the next section. 


\subsection{Ergodicity at energy $E$}

A projection $P\in \cC_N$ is {\em minimal} if it is non-zero, and if there is no non-zero projection $P'\in \cC_N$ such that
$P - P'$ is strictly positive.  In other words, if $P$ is minimial and $P'$ any other projection then either
they are not comparable or $P-P' \le 0$. Since for each $E\in \spec$, $P_E \in \cA_N \subset \cC_N$, every minimal projection $P$ in 
$\cC_N$  clearly satsfies $P_E - P \geq 0$ for some uniquely determined $E\in \spec$. To see this note that there must be an $E$ such that $PP_E=P_E=P \not= 0$ because the projections $P_E$ sum to  the identity.
Since $P$ is minimal, it must be comparable with $P_E$, for otherwise we we take $P_EPP_E$ which is again a projection and $PP_EP \le P$ contrary to our assumption of minimality. Hence there must be a unique $E$ such that
$P \le P_E$.  Since the algebra $\cC_N$ is evidently generated by its minimal projections, it is clear that $\cC_N = \cA_N$ if and only if $P_E$ is minimal in $\cC_N$
for each $E \in \spec$. This brings us to the following definition:

\begin{defi}[Ergodicity at energy $E$]\label{ergate} An ergodic collision specification $(\aC,U,\nu)$ is {\em ergodic at energy}
$E\in \spec$ in case  $P_E$ is minimal in $\cC_N$, and it is {\em fully ergodic} in case it is ergodic at  each $E\in \spec$. Equivalently, by what has been noted just above, $(\aC,U,\nu)$ is fully ergodic exactly when $\cC_N = \cA_N$.
\end{defi} 

The rest of this subsection is devoted to the characterization of the minimal projections dominated by $P_E$, $E\in \spec$, and hence of checking full ergodicity. 

\begin{lm}\label{seqlem} If a collision specification
 $(\aC,U,\nu)$ is ergodic, there is a finite sequence $\{\sigma_1,\dots,\sigma_s\}$ in  $\aC$ such that
$$\langle \psi_{e_k}\otimes \psi_{e_\ell}, U(\sigma_s)\cdots U(\sigma_2)U(\sigma_1) \psi_{e_m}\otimes \psi_{e_n}\rangle_{\H_2}\neq 0 \quad \iff \quad e_k+e_\ell = e_m + e_n\ .$$
\end{lm}

\begin{proof}
By ergodicity,
$$
\lim_{r\to\infty} \cQ^r( |\psi_{e_k} \otimes \psi_{e_\ell}\rangle \langle |\psi_{e_k} \otimes \psi_{e_\ell}| ) = \frac{1}{d} P_{e_k+e_\ell}
 $$
where $d$ is the dimension of the eigenspace of $H_2$ associated with the eigenvalue $e_k+e_\ell$.
(Both sides have unit trace.) Hence for some finite $s$, 
$$
\cQ^s( |\psi_{e_k} \otimes \psi_{e_\ell}\rangle \langle |\psi_{e_k} \otimes \psi_{e_\ell}| ) \geq  \frac{1}{2d} P_{e_k+e_\ell}
$$

Now  assume that $e_k+e_\ell = e_m + e_n$.
Taking the trace against $|\psi_{e_m} \otimes\psi_{e_n}\bk \psi_{e_m} \otimes\psi_{e_n}|$ yields
$$
 \int_{\aC\times \cdots \times  \aC} {\rm d}\nu(\sigma)|\langle \psi_{e_m}\otimes \psi_{e_n}, U(\sigma_s)\cdots U(\sigma_2)U(\sigma_1) \psi_{e_k} \otimes \psi_{e_\ell}\rangle|^2 {\rm d}\nu^{\otimes s}\geq  \frac{1}{2d} \ .
 $$
 The same computation also shows that if $e_k+e_\ell \not= e_m + e_n$  then 
$$
 \int_{\aC\times \cdots \times  \aC} {\rm d}\nu(\sigma)|\langle \psi_{e_m}\otimes \psi_{e_n}, U(\sigma_s)\cdots U(\sigma_2)U(\sigma_1) \psi_{e_k} \otimes \psi_{e_\ell}\rangle|^2 {\rm d}\nu^{\otimes s}= 0 \ ,
 $$
 and this proves the claim.
 \end{proof}


\begin{defi}[Adjacency]\label{adjac} Two indices $\bma,\bma'\in \cJ^N$ are {\em adjacent} in case
for some pair $(i,j)$, $1\leq i < j \leq N$, 
$$e_{\alpha_i} + e_{\alpha_j} = e_{\alpha'_i} + e_{\alpha'_j}$$
and for all $k\neq i,j$, $\alpha_k = \alpha'_k$. Two indices $\bma,\bma'\in \cJ^N$ are {\em equivalent} in case
there is a finite sequence $\{\bma^0,\dots, \bma^n\}$ such that $\bma^0 = \bma$, $\bma^n = \bma'$ and
$\bma^{j}$ and $\bma^{j-1}$ are adjacent for each $j =1,\dots, n$.  In this case we write $\bma \sim \bma'$. 
\end{defi}  

\begin{remark}\label{permu}
 Clearly if $\bma$ and $\bma'$ differ by a pair transposition; i.e., for some $1 \leq i < j \leq N$, 
 $\alpha_i = \alpha_j'$,  $\alpha_j = \alpha_i'$ and $\alpha_k = \alpha_k'$ for $k\neq i,j$, then 
 $\bma$ is adjacent to $\bma'$. 
Consequently, if $\bma'$ is related to $\bma$ by some permuations of the indices, then $\bma' \sim \bma$. 
\end{remark}

\begin{lm}\label{adjaclem} Let $A$ be self adjoint in $\cC_N$ with 
$A\Psi_\bma = \lambda_\bma \Psi_\bma$ for all $\bma\in \cJ^N$. Suppose that 
$(\aC,U,\nu)$ is ergodic. Then if $\bma$ and $\bma'$ are adjacent, $\lambda_\bma = \lambda_{\bma'}$, and more generally,
\begin{equation}\label{simlem2Z}
\bma \sim \bma' \quad \Rightarrow \quad \lambda_\bma = \lambda_{\bma'}\ .
\end{equation}
Conversely, if $\bma$ and $\bma'$ are not equivalent, there is some self-adjoint $A\in \cC_N$ for which
$\lambda_\bma \neq \lambda_{\bma'}$. 
\end{lm}

\begin{proof} Suppose that $\bma,\bma'\in \cJ^N$ are adjacent, with  $\alpha_k = \alpha'_k$  for all $k\neq i,j$,
By Lemma~\ref{seqlem}, 
 there is a finite sequence $\{\sigma_1,\dots,\sigma_s\}$ in  $\aC$ such that 
 $$\langle \Psi_{\bma'}, U_{i,j}(\sigma_s) \cdots U_{i,j}(\sigma_1)  \Psi_{\bma}\rangle \neq 0\ .$$
 Since $A$ is self adjoint and commutes with each $U_{i,j}(\sigma)$, 
 $$\langle A\Psi_{\bma'}, U_{i,j}(\sigma_s) \cdots U_{i,j}(\sigma_1)  \Psi_{\bma}\rangle =
 \langle \Psi_{\bma'}, U_{i,j}(\sigma_s) \cdots U_{i,j}(\sigma_1)  A\Psi_{\bma}\rangle\ ,$$
 and hence it must be that $\lambda_\bma = \lambda_{\bma'}$, and now
 \eqref{simlem2Z} follows easily. 
 
 It follows from Lemma~\ref{seqlem} that if $\bma$ and $\bma'$ are not equivalent, and therefore not adjacent, 
 $$\langle \Psi_{\bma}, U_{i,j}(\sigma)\Psi_{\bma'}\rangle = 0\ $$
 for all $i,j$ and all $\sigma$.  Expanding, $U_{i,j}(\sigma)\Psi_{\bma} = \sum_{\bmb} \langle \Psi_{\bmb}, 
 U_{i,j}(\sigma)
 \Psi_{\bma}\rangle$.  Evidently, $\langle \Psi_{\bmb}, 
 U_{i,j}(\sigma)
 \Psi_{\bma}\rangle =0$ unless $\bmb$ is adjacent to $\bma'$.
It follows that if $V$ is any finite  product of operators of the from $U_{i,j}(\sigma)$, 
$\langle \Psi_{\bma'}, V\Psi_{\bma}\rangle = 0$, unless $\bma \sim \bma'$.
Then by the definition of $\cP_{N,t}$, it follows that $\tr[ |\Psi_{\bma'}\bk\Psi_{\bma'}| \cP_{N,t}(
|\Psi_{\bma}\bk  \Psi_{\bma}|)] =0$ unless $\bma \sim \bma'$.

 By \eqref{approach}
 \begin{equation}\label{approachB}
\lim_{t\to\infty} \cP_{N,t} |\Psi_{\bma}\bk  \Psi_{\bma}| = \EC |\Psi_{\bma}\bk  \Psi_{\bma}|\ , 
\end{equation}
and then if $\bma$ and $\bma'$ are not equivalent, $(\EC |\Psi_{\bma}\bk  \Psi_{\bma}|)\Psi_{\bma'} = 0$. Thus,
$\EC |\Psi_{\bma}\bk  \Psi_{\bma}|$ is an operator in $\cC_N$ for which $\Psi_{\bma}$ and $\Psi_{\bma'}$ have distinct eigenvalues.
\end{proof} 

There is another useful way to phrase  Lemma~\ref{adjaclem}. A non-zero projection $P$ in an operator algebra is {\em minimal} in case whenever $Q$ is a non zero projection in the algebra such that $Q \leq P$, $Q = P$. 

\begin{lm}\label{adjaclemB} The minimial projections $P$  in $\cC_N$  satisfying $P \leq P_E$ for some $E\in \spec$
are precisely of the form
$$P = \sum_{\alpha \sim \alpha_0}|\Psi_\alpha\rangle\langle \Psi_\alpha|$$
for some $\alpha_0$ with $H_N\Psi_{\alpha_0} = E \Psi_{\alpha_0}$. 
That is, the minimal projections stand in one to one correspondence with the equivalence classes of the indices $\alpha$.  
\end{lm}

\begin{lm}\label{adjaclem2}  A collision specification is fully ergodic if and only if whenever $\bma,\bma'\in \cJ^N$ satisify
$E_\bma = E_{\bma'}$, then $\bma \sim \bma'$. 
\end{lm}

\begin{proof} Evidently $\cC_N \subset \cA_N$ if and only if whenever $P$ is a projection in $\cC_N$, it belongs to $\cA_N$, and this is the case if and only if for each $E\in \spec$, either $P_EP =0$ or $P_EP = P_E$. The claim now follows from Lemma~\ref{adjaclem}. 
\end{proof}

We next study some examples. The following notation will be useful. 

\begin{defi}[Occupancy function]\label{occ}
For each $j\in \cJ$, and each $N$, define the function $M_j$
on $\cJ^N$ with values in $\{0,1,\dots,N\}$ by
\begin{equation}\label{occf}
M_j(\bma) = \sum_{k=1}^N \delta_{\alpha_k,j} = m_j\ .
\end{equation}
The function is called the {\em occupancy function for the $j$th level}; it counts the number of times the 
eigenstate $\psi_j$ of $h$ appears as a factor in 
$$\Psi_\bma  = \psi_{\alpha_1}\otimes \cdots \otimes \psi_{\alpha_N}\ .$$
\end{defi}

The figures below illustrate the utility of the occupation function. Each figure is a histogram of the 
occupied levels in a system of 10 particles and 3 single particle states of energies $e_1 =1$, $e_2=2$, and $e_3 = 3$. 
Let $\bma_1$, $\bma_2$ and $\bma_3$ be any 3 elements of $\cJ^{10}$  with occupancy as indicated 
in the corresponding figure.
Then for each $j=1,2,3$, $\Psi_{\bma_j}$ is an eigenfunction of $H_{10}$ with eigenvalue $21$. 
However, $\bma_1\sim \bma_2$, while $\bma_3$ belongs to a second equivalence class.

\bigskip

\begin{tikzpicture}
\draw [fill=gray] (0,0) rectangle (1,4);
\node [below] at (0.5,0) {{\bf 1}};
\draw [fill=gray] (1,0) rectangle (2,2);
\draw [fill=black] (1,2) rectangle (2,4);
\node [below] at (1.5,0) {{\bf 2}};
\draw [fill=gray] (2,0) rectangle (3,3);
\node [below] at (2.5,0) {{\bf 3}};

\draw[help lines] (0,0) grid (3,5);

\draw [fill=gray] (5,0) rectangle (6,4);
\draw [fill=black] (5,4) rectangle (6,5);
\node [below] at (5.5,0) {{\bf 1}};
\draw [fill=gray] (6,0) rectangle (7,2);
\node [below] at (6.5,0) {{\bf 2}};
\draw [fill=gray] (7,0) rectangle (8,3);
\draw [fill=black] (7,3) rectangle (8,4);
\node [below] at (7.5,0) {{\bf 3}};

\draw[help lines] (5,0) grid (8,5);

\draw [fill=gray] (10,0) rectangle (11,3);
\node [below] at (10.5,0) {{\bf 1}};
\draw [fill=gray] (11,0) rectangle (12,3);
\node [below] at (11.5,0) {{\bf 2}};
\draw [fill=gray] (12,0) rectangle (13,4);
\node [below] at (12.5,0) {{\bf 3}};

\draw[help lines] (10,0) grid (13,5);

\node [below] at (1.5,-.5) {Fig. 1};

\node [below] at (6.5,-.5) {Fig. 2}; 

\node [below] at (11.5,-.5) {Fig. 3};
 
\end{tikzpicture}

\medskip

To see this, note that the only collisions that alter the histograms are those in which a pair with energies 
$e_1$ and $e_3$ becomes a pair in which both particles have energy $e_2$, or else the reverse of such a  collision.  
In fact $\bma_1$ and $\bma_2$ are adjacent, being connected by a single collision of this type. 
The squares that are changed in the histogram after such a collision are colored in black in Fig. 1 and 2.   
Evidently,  the occupancy of level 2 is even  after such a collision if and only if it was even before such a collision. 
Hence the parity of the occupnacy of the middle collumn cannot be changed by collisions.
This graphically illustates the fact that $\bma_3$ is not similar to $\bma_1$ or $\bma_2$.  
Evidently, this system of 10 particles with 3 single particle energy levels is not fully ergodic.  
We now consider some further exmaples.

\begin{exam}\label{irrat} Let $\cH = \C^n$, and let $h$ have eigenvalues $\{e_1,\dots, e_n\}$. If $n=2$, suppose that $e_1 \neq e_2$. For $n\geq 3$, suppose that $\{e_1,\dots,e_n\}$ is linearly independent over the rational numbers. Then $\spec$ consists of the real numbers of the form $E = \sum_{j=1}^n m_je_j$ where each $m_j$ is a non-negative integer, and $\sum_{j=1}^n m_j = N$. Because of the linear independence of the energies,
$$H_N\Psi_\bma = \left(\sum_{j=1}^n m_je_j\right)\Psi_\bma$$ if and only if for each $j=1,\dots,n$,
$$M_j(\bma) = m_j\ .$$
If follows that for all $\bma,\bma'$ with $E_\bma = E_{\bma'}$, $\bma$ and $\bma'$ are related by a permutation of indices, and then by Remark~\ref{permu}, $\bma\sim \bma'$.  It follows that for this choice of $h$, every ergodic collision specification is fully ergodic. 
\end{exam}

\begin{exam}\label{osc}  Let $\cH = \C^n$, and let $h$ have eigenvalues $e_j = j-1$, $j=1,\dots,n$. We may as well suppose that $n\geq 3$ since the case $n=2$ is covered by the previous example. In this case of evenly spaced eigenvalues, $n\geq 3$, there are many $\bma,\bma'$ that are adjacent, but unrelated by permutations. 
For each $m\in \N$, $m < n/2$, and $k$ with $m+1 \leq k \leq n-m$ we have obviously that $2e_k = e_{k-m} + e_{k+m}$. Thus, if $\bma$ and $\bma'$ are such that for some $1\leq i < j \leq N$ $\alpha_\ell = \alpha'_\ell$ for $\ell \neq i,j$, while
$$\alpha_i = \alpha_j = k \quad{\rm and} \quad \alpha'_i = k-m\ ,\ \alpha'_j = k+m$$
then $\bma$ and $\bma'$ are adjacent.  Taking into account the permuation invariance, it follows that if $\bma$ and $\bma'$ satisfy
$$M_k(\bma') = M_k(\bma)-2\quad{\rm  and}\quad  M_{k\pm m}(\bma') = M_{k\pm m}(\bma)+1$$
and $M_\ell(\alpha') = M_\ell(\alpha)$ for $\ell \neq k, k-m,k+m$, then $\bma'\sim \bma$. 

Now let us consider $n=3$. If $M_1(\bma) > 1$, we may lower $M_1(\bma)$ by $2$ through a collision of two particles of 
energy $e_2 =1$, producing a pair with energies $e_1 =0$ and $e_3 =2$. Doing this repeatedly, we arrive at an  $\bma'$
with either $M_1(\bma') =0$ or  $M_1(\bma') =1$, and with $\bma'\sim \bma$.  Since $E_{\bma'}$ is even if and only if
$M_1(\bma')$ is even,  it follows that whenever $E_\bma = E_{\bma'}$, then $\bma \sim \bma'$. Thus for this $h$ the ergodic collision specification $(\aC,U,\nu)$ is fully ergodic. 

Things are different for $n=4$. We may start from an arbitrary $\bma$, and consider collisions that decrease
$M_2(\bma)+M_3(\bma)$ while increasing $M_1(\bma)+M_4(\bma)$. We may continue doing so as long as either
$M_1(\bma)>1$ or  $M_2(\bma)>1$.  When the process stops, $M_2(\bma)$ and $M_3(\bma)$ are both either $0$ or $1$. 
The energy $E$ is given by
$$E_\bma = M_2(\bma) + 2M_3(\bma) + 3M_4(\bma)\ .$$
Evidently $E_\bma$ is of the form $3k+1$, $k\in \N$, if and only if $M_2(\bma) =1$ and $M_3(\bma) =0$. Likewise,  
$E_\bma$ is of the form $3k+2$, $k\in \N$,  if and only if $M_2(\bma) =0$ and $M_3(\bma) =1$. However
$E_\alpha =3k$, $k\in \N$ if and only if  either  $M_2(\bma) =M_3(\bma) =0$ or  $M_2(\bma) =M_3(\bma) =1$.
It is also clear that if $M_2(\bma) =M_3(\bma) =0$ and $\bma'\sim \bma$, then 
$M_2(\bma') =M_3(\bma') =0$.  Thus complete ergodicity is impossible in this case, but the problem only 
arrises when the energy $E$ is a multiple of $3$. 

One might expect things to get more complicated for $n=5$, but this is not the case: For $n=5$, any 
ergodic collision specification $(\aC,U,\nu)$ is fully ergodic on energy shells with $E/N$ not too close to either 
$0$ or $4$ when $N$ is sufficiently large.

To see this, consider an arbitrary $\bma'$, and then in a finite sequences of  steps one arrives at an equivalent $\bma$
such that $M_j(\bma) \in \{0,1\}$ for each $j=2,3,4$.  We may further suppose that among all such 
$\bma$ equivalent to $\bma'$, $\sum_{j=2}^4 M_j(\bma)$ is minimal. 

Under this minimality assumption, it is impossible that both $M_2(\bma)=1$ and  $M_4(\bma)=1$: 
If so, there is a collision that lowers these both to $0$, while rasing $M_3(\bma)$ by $2$.  
But this can be lowered again by $2$ in a collision that then raises both $M_1(\bma)$ and 
$M_5(\bma)$ by $1$. Thus, in two steps, one further lowers $\sum_{j=2}^4 M_j(\bma)$ by $2$. 

Next, suppose that $M_3(\bma) = M_4(\bma) >1$. If  $M_1(\bma) > 1$ (which will necessarily be the case for 
sufficiently large $N$), a collision can lower $M_1(\bma)$ and $M_2(\bma)$ each by $1$, and raising $M_2(\bma)$ by  
$2$. Then as above, in one more step $M_2(\bma)$ can be lowered again by $2$, while raising 
$M_1(\bma)$ and $M_5(\bma)$ each by $1$. Again, in two steps one has lowered $\sum_{j=2}^4 M_j(\bma)$ 
by $1$, which  is impssible under the minimality hypothesis. Thus under this hypothesis, if 
$M_3(\bma) =1$, $M_1(\bma)M_4(\bma) = 0$. The same reasoning shows that
$M_2(\bma)M_5(\bma) = 0$. 

Now if $E_\alpha/N$ is not too close to either $0$ or $4$, then necessarily $M_1(\bma)>1$ and $M_5(\bma) > 1$ whenver 
$M_j(\bma) \in \{0,1\}$ for each $j=2,3,4$.  Hence under the minimality hypothesis, the only possibilities for 
$M_j(\bma)$, $j=2,3,4$ are that  $M_j(\bma) = 0$ for all  $j=2,3,4$, in which case $E_\bma$ is of the form $4k$, $k\in \N$,
or else $M_j(\bma)=1$ for exactly one value of $j =2,3,4$, and is zero for all of the others. 
These respectivley correpond to the energies $4k + j-1$, $k\in \N$, $j=2,3,4$. 

Thus, for all $\eta>0$, and all $N$ sufficently large, if $E/N \in (\eta, 4-\eta)$, then there is full ergodicity at energy $E$
whenver $(\aC,U,\nu)$ is ergodic. 
\end{exam}

Since invariant densities for the QKME must belong to $\cC_N$, we have the following immediate conseqeuence of the characterization of 
$\cC_N$ obtained in this section:

\begin{thm}\label{longtime} Let $(\aC,U,\nu)$ be an  ergodic collision specification, and let $\cL_N$ 
be defined in terms of it as in (\ref{gendef}).  A density matrix $\varrho$ on $\H_N$ satisfies 
$\cL_N\varrho = 0$ if and only if 
it is a convex combination of normlaized minimal projections in $\cC_N$.
\end{thm} 

A density matrix $\varrho$ on $\H_N$ is a {\em product sate} if $\varrho = \rho_1\otimes \cdots \otimes \rho_N$ where each $\rho_j$ is a density matrix on $\H$.
A density matrix $\varrho$ on $\H_N$ is {\em separable} in case  $\varrho$ is a closed convex hull of the  product states.
A density matrix $\varrho$ on $\H_N$ is {\em entangled} in case is is not separable,

\begin{cl}[Separability of steady states] \label{sepcl}
Let $(\aC,U,\nu)$ be an  ergodic collision specification, and let $\cL_N$ 
be defined in terms of it as in (\ref{gendef}).  All density matrices $\varrho$ on $\H_N$ that satisfy
$\cL_N\varrho = 0$ are separable.
\end{cl}

\begin{proof}   Consider any  $\bma\in \cJ^N$. Then $\Psi_\bma  = \psi_{\alpha_1}\otimes \cdots \otimes \psi_{\alpha_N}$,
and evidently $|\Psi_{\bma}\bk  \Psi_{\bma}|$ is product state.  Since each minimal projection in $\cC_N$ is diagonal in the
$\{\Psi_{\bma}\}_{\bma\in \cJ^N}$ basis, the claim is a corollary of Theorem~\ref{longtime}.
\end{proof}

Consider any  ergodic collision specification $(\aC,U,\nu)$. Let $\cL_N$ be the associate quantum Markov semigroup generator, and let $\cP_{N,t} = r^{t\cL_N}$. Let $\varrho$ be any density matrix on $\H_N$, and let 
$$\varrho_{\infty} := \EC \varrho ,$$
which satisfies not on $\cL_N \varrho_\infty =0$, but by \eqref{approach},
\begin{equation}\label{approachC}
\lim_{t\to\infty} \cP_{N,t} \varrho = \varrho_{\infty}\ .
\end{equation}

It is clear that in the examples we have discussed with finite dimensional single particle space $\H$, it will be the case that
\begin{equation}\label{emlim}
\lim_{t\to\infty} S( \cP_{N,t} \varrho || \varrho_{\infty}) =0
\end{equation}
where for two density matrices $\rho,\sigma$, $S(\rho||\sigma) = \tr[\rho(\log \rho - \log \sigma)]$ is 
the Umegaki relative entropy of $\rho$ with respect to $\sigma$.  The rate at which the limit in  
\eqref{emlim} is attained is of interest. By a theorem of Lindblad \cite{L75}, derived as a consequence of the strong 
subadditivity of the quantum entropy proved by Lieb and Ruskai \cite{LR73},  the quantity $S( \cP_{N,t} \varrho || \varrho_{\infty})$ is monotone decreasing in $t$, and hence
$$D_N(\varrho) := -\frac{{\rm d}}{{\rm d}t}S( \cP_{N,t} \varrho || \varrho_{\infty})\bigg|_{t=0} \geq 0 \ .$$
A quantum analog of the Cercignani conjecture \cite{Cer} from classical kinetic theory would be that there exists a constant $c>0$ such that
$$\inf_{\varrho\in\Dens}\frac{D_N(\varrho)}{S(\varrho || \varrho_{\infty})} \geq c$$
uniformly in $N$. 

Likewise, there are various measures of entanglement for many-body systems, and the rates at which they decay to zero are of interest, These matters will be investigated in forthcoming work.

%
%
%
%
%
%
%
%

\section{The quantum Kac-Boltzmann equation}

\subsection{Propagation of chaos}

A density matrix $\varrho\in \Dens(\H_N)$ is {\em symmetric} in case it is invariant under the
canonical action of the permutation group on $\H_N$.   Here, the adjective  {\em symmetric} without further qualification will always have this meaning. 
For example, for each $E\in \spec$,
${\displaystyle \sigma_E = \frac{1}{{\rm dim}(\K_E)}P_E}$
 is symmetric. 
 
Given $\varrho\in \Dens(\H_N)$, let
$$\varrho^{(1)} = \tr_{2\dots N} \varrho  \qquad{\rm and \ more \ generally} \qquad  \varrho^{(k)} = \tr_{k+1\dots N}\varrho $$
where we take the partial trace of the the last $N-1$, respectively  $N-k$, factors  in $\H_N$. 
As usual, $\varrho^{(1)}$ is called the {\em single particle reduced density matrix} and 
$\varrho^{(k)}$ is called the {\em k-particle reduced density matrix}.

\begin{defi}[Chaoticity] Let $\rho$ be a density matrix on $\H$. A sequence $\{\varrho_N\}_{N\in \N}$ of 
symmetric density matrices on $\H_N$
is $\rho$-chaotic in case
$$\lim_{N\to\infty} \varrho^{(1)} = \rho  \qquad{\rm and}\qquad   \lim_{N\to\infty} \varrho^{(k)} = \otimes^k \rho\ .$$
\end{defi}

The point of the definition is that, as in the classical case, chaos is propagated, and propagation of chaos 
leads to a non-linear Boltzmann type equation. This is also true in the quantum case:

\begin{thm}\label{kacprop} Let $\{U(\sigma)\:\ \sigma\in \aC\}$ be an ergodic set of collision operators
and let $\nu$ be a given Borel probability measure on $\aC$.  Let $\cL_N$ be defined in terms of these as in 
(\ref{quankac9}). Then the semigroup $\cP_{N,t} = e^{t\cL_N}$  propagates chaos for all $t$  meaning that id 
$\{\varrho_N\}_{N\in \N}$ is a $\rho$-chaotic sequence, then for each $t$,
$\{\cP_{N,t} \varrho_N\}_{N\in \N}$ is a $\rho(t)$-chaotic sequence for some $\rho(t) = \lim_{N\to\infty}(\cP_{N,t} \varrho_N)^{(1)}$, where in particular this limit of the  one-particle maginal exists and is a density matrix.  
\end{thm} 

Before beginning the proof, we explain the strategy, which follows that of the orignal argument of Kac as refined by McKean. 
Consider an operator $A$ of the form $B_k \otimes \one_{\H_{N-k}}$ where  $1\leq k < N$, and $B_k\in \cB(\H_k)$.
We are interested in estimating
$\tr[ \varrho \cL_N B]$
where $\varrho$ is symmetric. 
First note that 
\begin{eqnarray}\label{lnn}
\cL_N(B_k \otimes \one_{\H_{N-k}}) 
&=& \frac{2}{N-1} \sum_{1 \leq i < j \leq k} (\cQ_{i,j} - \one_{\H_k})(B_k) \otimes \one_{\H_{N-k}}\nonumber\\
&+&  \frac{2}{N-1} \sum_{1 \leq i \leq k, j> k} (\cQ_{i,j} - \one_{\H_k})(B_k) \otimes \one_{\H_{N-k}} \ .
\end{eqnarray}
For $N$ much larger than $k$, there are many more terms in the second sum on the right in \eqref{lnn} than in the first. Moreover, when taking the expectation against $\varrho$ with $\varrho$ symmetric, each of these terms makes the exact same contribution;
for $j > k$, let $\pi$ be the pair permuation sending $j$ to $k+1$ and $k+1$ to $j$. Then since
$U_\pi \varrho U_\pi^* = \varrho$, it follows from \eqref{Upidef4} and the fact that $\pi = \pi^{-1}$ that 
\begin{eqnarray*}
\tr[\varrho  (\cQ_{i,j} - \one_{\H_k})(B_k) \otimes \one_{\H_{N-k}}] &=&
\tr[U_\pi \varrho U_\pi^*  (\cQ_{i,j} - \one_{\H_k})(B_k) \otimes \one_{\H_{N-k}}] \nonumber\\
&=&
\tr[ \varrho U_\pi^*  (\cQ_{i,j} - \one_{\H_k})(B_k \otimes \one_{\H_{N-k}})U_\pi] \nonumber\\
&=& \tr[\varrho  (\cQ_{\pi(i),\pi(j)} - \one_{\H_k})(B_k) \otimes \one_{\H_{N-k}}] 
 \nonumber\\
&=& \tr[\varrho  (\cQ_{i,k+1} - \one_{\H_k})(B_k) \otimes \one_{\H_{N-k}}] \ .
\end{eqnarray*}
Therefore,
\begin{eqnarray}\label{lnn2}
\tr[\varrho \cL_N(B_k \otimes \one_{\H_{N-k}})]
&=& \frac{2}{N-1} \sum_{1 \leq i < j \leq k} \tr[ \varrho (\cQ_{i,j} - \one_{\H_k})(B_k) \otimes \one_{\H_{N-k}}]\nonumber\\
&+&  \frac{(N-k)}{N-1} 2\sum_{1 \leq i \leq k} \tr[\varrho (\cQ_{i,k+1} - \one_{\H_k})(B_k) \otimes \one_{\H_{N-k}}]
\end{eqnarray}
For large $N$, the second term will turn out to be the main term, as is almost clear from the factor of $1/(N-1)$ in front of the first term. 

The fact that the first term is negligible in the limit $N\to\infty$ has a probabilistic interpretation emphasized by Kac. Consider
 the ``collision histories'' of the first $k$ particles over some fixed time interval $[0,t]$. Then the probability that any of these particles collide with each other, or even collide with any particle that has already  collided with any of the first $k$ particles, is vanishingly small in the limit $N\to \infty$. Then, as we shall see, without recollisions, there is no mechanism for generating correlations.

To efficiently work with the second term, and to see why it does not generate correlation, we introduce, following McKean, the operators $\Gamma_k:
\cB(\cH_k) \rightarrow \cB(\cH_{k+1})$ by
$$
\Gamma_k (B_k) = 2  \sum_{i=1}^k (\cQ_{i,k+1} - \one_{\cH_{k+1}})(B_k \otimes \one_\cH) \ ,
$$

Now let $\{\varrho_N\}_{N\in \N}$ be a $\rho$-chaotic seqence. Then, 
$$\lim_{N\to\infty} \tr[\varrho_N \cL_N(B_k \otimes \one_{\H_{N-k}})] =  \tr[ \otimes^{k+1}\rho \Gamma_k(B_k)]\ .$$
A closer analysis, carried out below,  of the decomposition in \eqref{lnn2} will show that for all $\ell \in N$, 
$$\lim_{N\to\infty} \tr[\varrho_N \cL_N^\ell (B_k \otimes \one_{\H_{N-k}})] =  \tr[
 (\otimes^{k+\ell}\rho) \Gamma_{k+\ell-1} \cdots  \Gamma_k(B_k)]\ .$$

Furthermore, it will be shown that the power series for  $\tr[\varrho_N e^{t\cL_N}(B_k\otimes \one_{\H_{N-k}})]$ converges {\em uniformly in $N$ for small $t$}.   Once this is shown, the propagation of chaos will follow from a term by term 
analysis of the power series.  
The key for this is the observation of McKean  that $\Gamma_k$ is a ``twisted'' derivation in the following sense:  Let $1\leq j < k$,
$X_j\in \cB(\H_j)$ and $Y_{k-j}\in \cB(\H_{k-j})$. Let $\pi$ denote the pair permuation on $\{1,\dots, k+1\}$ such that
$\pi(j+1) = k+1$ and $\pi(k+1) = j+1$. Then 
\begin{equation}\label{longver}
\Gamma_k(X_j\otimes Y_{k-j}) = 
[U_\pi (\Gamma_{j}(X_j)\otimes \one_{\H_{k-j}})U_\pi^* ]\, \one_{\H_{j}} \otimes 
Y_j \otimes \one_\H + X_j\otimes \Gamma_{k-j}Y_{k-j}\ .
\end{equation}
With a slight abuse of notation, we can write this more simply as
\begin{equation}\label{shortver}
\Gamma_k(X_j\otimes Y_{k-j}) =  (\Gamma_j X_j)\otimes Y_{k-j} +  X_j\otimes \Gamma_{k-j}Y_{k-j}\ ,
\end{equation}
but then with the understanding that $\Gamma_j X_j$ acts on the first $j$ factors of $\H$ together with the 
$k+1$st factor, as is made clear by the permuations in the longer form. 

For example, consider the case $k=2$, which is the most important for our binary collison model. Let $B_2 = X_1\otimes Y_2$.
Then for $\ell\in \N$, supressing permuations from our notation as in the passage from \eqref{longver} to \eqref{shortver}, with $N > \ell+2$, so that $\Gamma_\ell (X_1\otimes Y_2)$ is defined,
\begin{eqnarray*}
&&\lim_{N\to\infty} \tr[\varrho_N \cL_N^\ell (X_1\otimes Y_2 \otimes \one_{\H_{N-2}})] \nonumber\\ &=&  
\tr[\otimes^{2+\ell}\rho \Gamma_{2+\ell-1} \cdots  \Gamma_2(X_1\otimes Y_2)]\nonumber\\
&=&   \sum_{\ell_1+\ell_2= \ell} \frac{\ell!}{\ell_1! \ell_2!}
\tr[\otimes^{2+\ell}\rho (\Gamma_{\ell_1}\cdots \Gamma_1)X_1 \otimes (\Gamma_{\ell_2}\cdots \Gamma_1)Y_2]\nonumber\\
&=&   \sum_{\ell_1+\ell_2= \ell} \frac{\ell!}{\ell_1! \ell_2!}
\tr[\otimes^{1+\ell_1}\rho (\Gamma_{\ell_1}\cdots \Gamma_1)X_1]  \tr[\otimes^{1+\ell_2}\rho  (\Gamma_{\ell_2}\cdots \Gamma_1)Y_2]\nonumber\\
&=&\lim_{N\to\infty} \sum_{\ell_1+\ell_2= \ell} \frac{\ell!}{\ell_1! \ell_2!} 
\tr[\varrho_N \cL_N^{\ell_1} (X_1\otimes \one_{\H_{N-1}})]
\tr[\varrho_N \cL_N^{\ell_1} (Y_2\otimes \one_{\H_{N-1}})]
 \end{eqnarray*}
 
After the reduction that permits us take limits term by term, we conclude 
$$\lim_{N\to\infty}[(e^{t\cL_N}\varrho_N )X_1\otimes Y_2 \otimes \one_{\H_{N-2}}] = 
\lim_{N\to\infty}[(e^{t\cL_N}\varrho_N) X_1 \otimes \one_{\H_{N-1}}] [(e^{t\cL_N}\varrho_N) Y_1 \otimes \one_{\H_{N-1}}]\ .$$
Since $X_1,Y_2\in \cB(\H)$ are arbitrary, this means that
$$\lim_{N\to\infty}(e^{t\cL_N}\varrho_N )^{(2)} = \lim_{N\to\infty}(e^{t\cL_N}\varrho_N )^{(1)}\otimes (e^{t\cL_N}\varrho_N )^{(1)}\ .$$
This  shows that $\{e^{t\cL_N}\varrho_N \}_{N\in\N}$ is $(e^{t\cL_N}\varrho_N )^{(1)}$-chaotic. 

Having explained the strategy and the key role of McKean's derivation property, it remains to provide the 
estimates that permit it to be carried out.

Define $G_k: \cB(\cH_k) \rightarrow \cB(\cH_{k+1})$ by
$$
G_k(B_k) =  \frac{2}{N-1} \sum_{1\le i < j \le k}
 (\cQ_{i,j} - \one_{\cH_k})(B_k) \otimes \one_\cH + \frac{N-k}{N-1} \Gamma_k (B_k)\ .
$$
\begin{lm} \label{technicalestim}  For all $k$, 
\begin{equation}\label{gell}
\Vert G_{\ell+k} \cdots  G_k(B_k) \Vert_{L(\cH_{k+\ell+1})} \le 4^{\ell+1} (\ell+k) \cdots k \Vert B_k \Vert_{\cB(\cH_k)}\ .
\end{equation}
For $\ell < N-k$,  so that $\Gamma_\ell B_k$ is well-defined, we have the following estimates:
\begin{equation}\label{gammaell}
\Vert \Gamma_{\ell+k} \cdots  \Gamma_k(B_k) \Vert_{L(\cH_{k+\ell+1})} \le 4^{\ell+1} (\ell+k) \cdots k \Vert B_k \Vert_{\cB(\cH_k)} \ .
\end{equation}
\begin{equation} \label{gelldiffgammaell}
\Vert G_{\ell+k} \cdots  G_k(B_k) - \Gamma_{\ell+k} \cdots  \Gamma_k(B_k) \Vert_{L(\cH_{k+\ell+1})} \le \frac{C_k 4^\ell \ell^{k+2}}{N-1} \ell!  \Vert B_k \Vert_{\cB(\cH_k)}  \ ,
\end{equation}
where $C_k$ is some constant that depends on $k$.
\end{lm}
\begin{proof}

Using the elementary estimate $\Vert \cQ_{i,k+1}(B_k\otimes \one_\cH)\Vert_{L(\cH_{k+1})}\le \Vert B_k \Vert_{\cB(\cH_k)}$,
the estimate on $\Vert G_k(B_k) \Vert_{L(\cH_{k+1})}$ is then a consequence of 
$$
k\frac{4N -3k-1}{N-1} \Vert B_k \Vert_{\cB(\cH_k)} \le 4k  \Vert B_k \Vert_{\cB(\cH_k)} \ .
$$
Again using the elementary estimate $\Vert \cQ_{i,k+1}(B_k\otimes \one_\cH)\Vert_{L(\cH_{k+1})}\le \Vert B_k \Vert_{\cB(\cH_k)}$ we find
$$
\Vert \Gamma_k(B_k)\Vert_{L(\cH_{k+1})} \le 4k \Vert B_k\Vert_{L(\cH_{k})}
$$
from which \eqref{gammaell} follows immediately. 

A bit trickier is the proof of \eqref{gelldiffgammaell}.
Using the estimate on $G_k$ and $\Gamma_k$ one easily finds that
\begin{equation}\label{gdiffgamma}
\Vert G_k(B_k) - \Gamma_k(B_k)  \Vert_{L(\cH_{k+1})} \le \frac{2k(k-1)+(k-1)4k}{N-1}  \Vert B_k \Vert_{\cB(\cH_k)} \le \frac{6k^2}{N-1}  \Vert B_k \Vert_{\cB(\cH_k)} \ .
\end{equation}
We write the telescoping sum
\begin{eqnarray*}
&&G_{\ell+k} \cdots  G_k(B_k) - \Gamma_{\ell+k} \cdots  \Gamma_k(B_k)\\
&=&
(G_{\ell+k}-\Gamma_{\ell+k}) G_{\ell+k-1} \cdots G_k(B_k) \\
&+& \sum_{m=1}^{\ell-1} \Gamma_{\ell+k} \cdots \Gamma_{\ell+k-m+1}(G_{\ell+k-m}- \Gamma_{\ell+k-m}) G_{\ell+k-m-1} \cdots G_k(B_k)\\
& +& \Gamma_{\ell+k} \cdots \Gamma_{k+1}(G_k-\Gamma_k)(B_k) \ .\\
\end{eqnarray*}
Using \eqref{gell}, \eqref{gammaell} and \eqref{gdiffgamma}, 
the norm of the right side can be estimated by
$$
\frac{6}{N-1} 4^\ell (\ell+k) \cdots k (\ell+1) [ k+\ell/2] \Vert B_k \Vert_{\cB(\cH_k)}
$$
$$
\le \frac{6}{N-1} 4^\ell (\ell+1) \frac{ (\ell+k+1)!} {(k-1)!}\Vert B_k \Vert_{\cB(\cH_k)} \ .
$$
Using Stirling's formula we find that
$$
\frac{ (\ell+k+1)!}{\ell!} \approx \left(1+\frac{k+1}{\ell}\right)^{\ell+k+3/2}\ell^{k+1} e^{-k-1}
$$
which is bounded by  $D_k \ell^{k+1}$ where $D_k$ is some constant that depends on $k$.
This proves \eqref{gelldiffgammaell}.
\end{proof}

\begin{proof}[Proof of Theorem \ref{kacprop}]
The proof is now a word by word translation of the one given by McKean \cite{McKean} for the classical case. We start by writing
${\displaystyle 
\cP_{N,t} = \sum_{\ell=0}^\infty \frac{t^\ell}{\ell!} \cL_N^\ell}$
so that
\begin{equation}\label{series}
\tr(\cP_{N,t}\varrho_N (B_k \otimes^{N-k} \one_\cH )) = \sum_{\ell=0}^\infty \frac{t^\ell}{\ell!} \tr (\varrho_N \cL_N^\ell (B_k \otimes^{N-k} \one_\cH))
\end{equation}
using that $\cP_{N,t}$ is self-adjoint.  By what has been explained above, 
%
%
%
%
%
%
%
%
%
%
%
%
for $\ell \ge 1$
$$
\tr (\rho_N \cL_N^\ell (B_k \otimes^{N-k} \one_\cH)) = \tr (\rho_N  G_{k+\ell-1}G_{k+\ell-1} \cdots G_k (B_k) \otimes^{N-k-\ell -1} \one_\cH))\ .
$$
Now using \eqref{gell}, we see that 
$$\frac{t^\ell}{\ell!} |\tr (\varrho_N \cL_N^\ell (B_k \otimes^{N-k} \one_\cH))|  \leq 
\frac{t^\ell}{\ell!} 4^{\ell+1} (\ell+k) \cdots k \Vert B_k \Vert_{\cB(\cH_k)}\ .$$
Using Stirling's formula we find that
$$
\frac{ (\ell+k)!}{\ell! k!} \approx \frac{1}{k!}\left(1+\frac{k}{\ell}\right)^{\ell+k+1/2}\ell^{k} e^{-k}  \leq 
\frac{1}{k!}\left(e^k\right)^{k-1/2}\ell^{k} \ .
$$
Then since $\sum_{\ell=0}^\infty t^\ell 4^\ell \ell^k$ converges uniformly for $|t| < 1/4$. we see that the series on the right in \eqref{series} converges at a rate independent of $N$ for all $0 \leq t < 1/4$. 

We are now in a position to take the limit $N\to\infty$ term by term.
Using \eqref{gelldiffgammaell}  for $N > \ell +k$, we obtain the estimate 
$$
|\tr (\rho_N \cL_N^\ell (B_k \otimes^{N-k} \one_\cH)) -  \tr (\rho_N  \Gamma_{k+\ell-1} \cdots \Gamma_k (B_k) \otimes^{N-k-\ell -1} \one_\cH))|
$$
$$
\le \frac{C_k4^\ell}{N-1} \ell^{k+2} \ell!\Vert B_k \Vert_{\cB(\cH_k)}  \ .
$$
Thus,
$$
\Big |\sum_{\ell=0}^{N-k} \frac{t^\ell}{\ell!} \tr (\rho_N \cL_N^\ell (B_k \otimes^{N-k} \one_\cH))
-\sum_{\ell=0}^{N-k} \frac{t^\ell}{\ell!}  \tr (\rho_N  \Gamma_{k+\ell-1} \cdots \Gamma_k (B_k) \otimes^{N-k-\ell -1} \one_\cH))\Big |
$$
$$
\le\sum_{\ell=0}^{N-k} \frac{t^\ell}{\ell!} \Big | \tr (\rho_N \cL_N^\ell (B_k \otimes^{N-k} \one_\cH))
- \tr (\rho_N  \Gamma_{k+\ell-1} \cdots \Gamma_k (B_k) \otimes^{N-k-\ell -1} \one_\cH)) \Big |
$$
$$
\le \frac{C_k}{N-1}\sum_{\ell=0}^{N-k} (4t)^\ell \ell^{k+2} \Vert B_k \Vert_{\cB(\cH_k)} 
\le   \frac{C_k}{N-1} \sum_{\ell=0}^\infty (4t)^\ell \ell^{k+2} \Vert B_k \Vert_{\cB(\cH_k)}  \ .
$$
provided that $4t< 1$. Again, the reason for summing up to $N-k$ only is that for $\ell>N-k$ the expression $\Gamma_{k+\ell}$ is not defined. By what we have proved at the beginning, however, the tail of the sum makes a negligible contribution in the large $N$ limit.

\end{proof}


\subsection{Quantum Convolution and the Quantum Kac-Boltzmann equation}

In the  classical case, the single particle density satsifies the Kac-Boltzmann equation. In the quantum case, it satisfies a quantum analog of the Kac-Boltzmann equation, as we now explain. 

\begin{lm}\label{12lem} Let $\varrho(t) = e^{t\mathcal{L}}\varrho_0$ for a symmetric density matrix $\varrho_0$
on $\H_N$. Then the one and two particle reduced density matrices of $\varrho(t)$  are related by 
\begin{equation}\label{redu}
\frac{{\rm d}}{{\rm d}t} \varrho^{(1)}(t) = 
2\tr_2\left[ \int_{\aC}\dd \nu(\sigma)  U(\sigma) [\varrho^{(2)}(t)] U^*(\sigma)
-\varrho^{(2)}(t)  \right]  \ .
\end{equation}
\end{lm} 

\begin{proof}
\begin{eqnarray}\frac{{\rm d}}{{\rm d}t} \rho^{(1)}(t) &=& \tr_{2,\dots,N}\left[ \mathcal{L}_N\varrho(t)\right]\nonumber\\
&=& \tr_{2,\dots,N}\left[ 
\frac{2}{N-1}\sum_{j=2}^N \int_{\aC}\dd \nu(\sigma)  \left[U_{1,j}(\sigma) \varrho(t) U^*_{1,j} (\sigma)   - 
\varrho(t)\right]\right]\nonumber\\
&=& \tr_{2,\dots,N}\left[ 
2 \int_{\aC}\dd \nu(\sigma)  \left[U_{1,2}(\sigma) \varrho(t) U^*_{1,2} (\sigma)   - 
\varrho(t)\right]\right]\nonumber
\end{eqnarray}
where  the symmetry was used in the last step. Now take $\tr_{3,\dots,N}$ to obtain the result. 
\end{proof}

Observe that if $\varrho^{(2)}(t) = \varrho^{(1)}(t)\otimes \varrho^{(1)}(t)$, then (\ref{redu}) would reduce to a closed equation for $\varrho^{(1)}(t)$:
\begin{equation}\label{reduA}
\frac{{\rm d}}{{\rm d}t} \varrho^{(1)}(t) = 
2\tr_2\left[ \int_{\aC}\dd \nu(\sigma) U(\sigma) [ \varrho^{(1)}(t)\otimes \varrho^{(1)}(t)] U^*(\sigma)]\right]
-\varrho^{(1)}(t)    \ .
\end{equation}
This brings us to the following definition:
\begin{defi}[quantum Wild convolution operator]\label{wilddef}  Let $(\aC,U,\nu)$ be a collision specification, The corresponding  {\em  quantum Wild  convolution} is the bilinear from $\cT(\H)\times \cT(\H)$ to  $\cT(\H)$
sending $(A,B)$ to $A\star B$ where
\begin{equation}\label{gain}
A\star B  =   \tr_2 \left[\int_{\aC}{\rm d}\nu(\sigma)U(\sigma) [A\otimes B] U^*(\sigma)\right] = \tr_2[\cQ(A\otimes B)]
\end{equation}
where $\cQ$ is the operator defined (\ref{echeck2}).
\end{defi}


Note that for any $A,B\in \cT(\H)$, 
\begin{eqnarray*}
\tr[A \star B] &=&  \tr_{1,2} \left[\int_{\mathcal{S}} U(\sigma) [A\otimes B] U^*(\sigma)\dd \sigma\right]\\
&=& \tr_{1,2} \left[\int_{\aC}{\rm d}\nu(\sigma) \tr_{1,2} [A\otimes B] \right] = \tr[A]\tr[B]\ .
\end{eqnarray*}
Furthermore if $A$ and $B$ are non-negative operators, then $A\star B$ is also non-negative. 
In particular, for $\rho\in \Dens(\H)$,  $\rho\star\rho\in \Dens(\H)$, and consequently
 
\begin{equation}\label{colop1}
\frac{{\rm d}}{{\rm d}t}  \rho(t) = 2(\rho(t) \star \rho(t) - \rho(t) )\ .
\end{equation}
is an evolution equation in $\Dens(\H)$. To see that is has unique global solutions, write it in the equivalent form
${\displaystyle \frac{{\rm d}}{{\rm d}t}  \left(e^{2t} \rho(t)\right)  = e^{2t} \rho(t) \star \rho(t)  }$, which, given the initial state $\rho_0$, can be integrated to obtain
\begin{equation}\label{colop2}
 \rho(t) = e^{-2t}\rho_0 + \int_0^t e^{2(s-t)} \rho(s) \star \rho(s){\rm d}s   \ .
\end{equation}
The equation (\ref{colop2}) may be solved by iteration as shown by McKean in the classical case, and the unique solution may be represented as a convergent sum over ``McKean graphs''.

\begin{defi}[The Quantum Kac-Boltzmann Equation] 
The Quantum  Kac Boltzmann Equation (QKBE)  is the evolution equation $\Dens(\H)$ given by
\begin{equation}\label{colop6}
\frac{{\rm d}}{{\rm d}t}  \rho(t) = 2(\rho(t) \star \rho(t)  -\rho(t))\ .
\end{equation}
\end{defi} 

\begin{thm}  Suppose that  $\{\varrho_N(0)\}_{N\in \N}$ is  $\rho(0)$-chaotic, and that for each $N$,
$\varrho_N(t) = \exp(t\mathcal{L}_N )\varrho_N(0)$ for all $t> 0$. Then $\rho(t)$  satisfies the 
Quantum Kac-Boltzmann Equation.
\end{thm} 

\begin{proof}  This is an immediate consequence of Theorem~\ref{kacprop} and Lemma~\ref{12lem}, and what we have said above about solutions of the QKBE. 
\end{proof}

\subsection{The quantum Wild convolution}

The operator $\cQ$ defined in Lemma~\ref{convlem} is the same as the operator $\cQ$ that arises in the Wild convolution. 
In Example~\ref{exam2} we saw that $\cQ = {\rm E}_{\cA_2}$.  In general, when $(\aC,U,\nu)$ is ergodic,
$\lim_{k\to\infty} \cQ^k = {\rm E}_{\cA_2}$, but if $\nu$ is uniform enough,as in Example~\ref{exam2},  it may not be necessary to take the limit. 
When $\cQ = {\rm E}_{\cA_2}$ it is  possible to give an explicit formula for the quantum Wild convolution $A \star B $ in terms of the spectral decomposition of $h$ using the formula (\ref{EAform}) for ${\rm E}_{\cA_2}$.

\begin{lm}\label{explicit}  Let $(\aC,U,\nu)$ be a collision specification  such that $\cQ$, as defined in terms of 
$(\aC,U,\nu)$ in  (\ref{echeck2}), satisfies  $\cQ =  {\rm E}_{\cA_2}$. For 
$E\in {\rm Spec}(H_2)$ and $j\in \cJ$, define 
\begin{equation}\label{normalize}
 M_E = \sum_{j\in \cJ} 1_{{\rm Spec}(h)} (E - e_j)\ .
\end{equation}
Then $M_E < \infty$ for all $E\in {\rm Spec}(H_2)$, and for all  $A,B\in \cT(\H)$,
\begin{equation}\label{wild2}
A\star B  = \sum_{i,j,k\in \cJ} 
\langle \psi_i, A\psi_i\rangle  \langle \psi_k, B\psi_k\rangle \frac{1_{{\rm Spec}(h)} (e_i+e_k - e_j)}    {M_{e_i+e_k}}|\psi_j \bk \psi_j|\ .
\end{equation}

\end{lm} 

\begin{proof}  The fact that $M_E < \infty$ is a direct consequence of the compactness of the resolvent of $h$. Write 
${\displaystyle A = \sum_{i,j\in \cJ} a_{i,j} | \psi_i\bk \psi_j| }$ and  ${\displaystyle 
B = \sum_{k,\ell\in \cJ} b_{k,\ell
} | \psi_k\bk \psi_\ell|}$ 
in terms of an eigenbasis of $h$ so that 
 ${\displaystyle A\otimes B = \sum_{i,j k,\ell\in \cJ}a_{i,j}b_{k,\ell} |\psi_i\otimes \psi_k \bk \psi_j\otimes \psi_\ell|}$.
 Then
$$\cQ(A\otimes B)  = \sum_{i,j k,\ell\in \cJ}a_{i,j}b_{k,\ell} \cQ( |\psi_i\otimes \psi_k \bk \psi_j\otimes \psi_\ell|) =
\sum_{i,j k,\ell\in \cJ}a_{i,j}b_{k,\ell} {\rm E}_{\cA_2}( |\psi_i\otimes \psi_k \bk \psi_j\otimes \psi_\ell|)\ .$$
By (\ref{EAform})
$$ {\rm E}_{\cA_2}( |\psi_i\otimes \psi_k \bk \psi_j\otimes \psi_\ell|) = \tr[ P_E |\psi_i\otimes \psi_k \bk \psi_j\otimes \psi_\ell|] \sigma_E  = \delta_{i,j}\delta_{k,\ell}\sigma_{e_i+e_k}\ .$$
Therefore,
${\displaystyle 
\cQ(A\otimes B) =  \sum_{i,k\in \cJ} a_{i,i}b_{k,k} \sigma_{e_i+e_k}}$.
For each $E\in {\rm Spec}(H_2)$,  
$$\tr_2[\sigma_{E}] = \frac{1}{M_{E}} \sum_{j\in \cJ} 1_{{\rm Spec}(h)} (E - e_j)    |\psi_j \bk \psi_j|\ ,$$
and combining this with the expression for $\cQ(A\otimes B)$ yields the result. 
\end{proof}

The formula (\ref{wild2}) simplifies further if $h$ satisfies a certain non-degeneracy condition, as shown in the next lemma. We first define an operator that will appear in the simplified formula.

\begin{defi} [Diagonal projection] Let $\{\psi_j\}$ be an orthonormal basis of $\H$ consisting of eigenvectors of $h$. The operator $\cD_h$ on $\cT_2(\H)$ is defined by
$$\cD_h(A) = \sum_{j\in \cJ} \langle \psi_j, A\psi_j\rangle |\psi_j\bk \psi_j|\ .$$
$\cD_h$ is the orthogonal projection onto the commutative {\em diagonal algebra} of that is generated by the spectral projections of $h$.
\end{defi}

\begin{lm}\label{nondeg}  Let $(\aC,U,\nu)$ be a collision specification  such that $\cQ$, as defined in terms of it in  (\ref{echeck2}), satisfies  $\cQ =  {\rm E}_{\cA_2}$.  Suppose that for all $i,j,k,\ell\in \cJ$, if $e_i+ e_j = e_k+ e_\ell$,
then either $i=k$ and $j=\ell$ or $i = \ell$ and $j=k$.  Then  for all $A,B\in \cT(\H)$,
$$A\star B = \frac12\left(  \tr[A] \cD_h(B)  + 
\tr[B] \cD_h(A)\right)\ .$$
In particular, for $\rho\in \Dens(\H)$,
$$\rho \star \rho = \cD_h(\rho)\ .$$
\end{lm}

\begin{proof} Under the non-degeneracy condition, for all $E\in {\rm Spec}(H_2)$, $M_E =1$ in case $E = 2e_j$ for some $j\in \cJ$, or else  $E = e_i+e_j$ for some $i\neq j$, and then $M_E  = 2$.
Therefore, (\ref{wild2}) becomes 
\begin{equation}\label{wild2C}
A\star B  = \frac12 \sum_{i,k\in \cJ} 
\langle \psi_i, A\psi_i\rangle  \langle \psi_k, B\psi_k\rangle (|\psi_j \bk \psi_j| + |\psi_k\bk \psi_k|)\ .
\end{equation}
\end{proof} 

When the conditions of Lemma~\ref{nondeg} are satisfied, the QKBE reduces to the {\em linear} evolutions equation
\begin{equation}\label{colop6B}
\frac{{\rm d}}{{\rm d}t}  \rho(t) = 2(\cD_h\rho(t)  -\rho(t))\ .
\end{equation}

Since the conditions of Lemma~\ref{nondeg} are satisfied for the collision specification in Example~\ref{exam2}, 
the QKBE for this example is linear.  When the non-degeneracy condition in Lemma~\ref{nondeg} 
is not satisfied, then
the map $\rho \mapsto \rho\star \rho$ need not be linear on $\Dens(\H)$, as Example~\ref{deg} 
below shows. However, when $\cQ = {\rm E}_{\cA_2}$, Lemma\ref{explicit} shows that 
$A \star B = \cD_h(A) \star \cD_h(B)$, so that the quantum Wild convolution of 
$\rho_1,\rho_2\in \Dens(\H)$ in this case is really a convolution of the classical 
probability vectors on the diagonals of $\rho_1$ and $\rho_2$.   

However, when $\cQ \neq {\rm E}_{\cA_2}$, but only $\lim_{k\to\infty}\cQ^k = {\rm E}_{\cA_2}$, 
the quantum Wild convolution is a more interesting quantum operation.

\begin{exam}\label{ncexam}  In this example, let $(\aC,U,\nu)$ be the collision specification from Example~\ref{exam2A}.   Let $a,b\in [0,1]$ and let $w,z\in \C$ satisfy $|z|,|w| \leq 1$ so that with
$$\rho_1 = \left[\begin{array}{cc} a & z\\ \overline{z} & 1-a\end{array}\right] \qquad{\rm and}\qquad 
\rho_2 = \left[\begin{array}{cc} b & w\\ \overline{w} & 1-b\end{array}\right] \ ,$$
$\rho_1$ and $\rho_2$ are generic elements of $\Dens(\C^2)$.   Then using the basis from Examples~\ref{exam2}
to identify $\C^2\otimes \C^2$ with $\C^4$, 
\begin{equation}\label{tpro}
\rho_1\otimes\rho_2 = \left[\begin{array}{cccc}
ab & zb & aw & zw\\
\overline{z}b & b(1-a) & \overline{z}w & w(1-a)\\
a\overline{w} & z\overline{w} & a(1-b) & z(1-b)\\
\overline{zw} & (1-a)\overline{w} & (1-b)\overline{z} &(1-a)(1-b)\end{array}\right]\ .
\end{equation}
Then by (\ref{notproj}),
\begin{equation*}
\cQ(\rho_1\otimes\rho_2) = 
\left[\begin{array}{cccc}
ab & \frac18 zb & \frac18 aw & \frac12 zw\\
\frac18\overline{z}b & \frac12 (a+b)-ab & 0 & \frac14 w(1-a)\\
\frac18 a\overline{w} & 0 & \frac12 (a+b)-ab & \frac14 z(1-b)\\
\frac12 \overline{zw} & \frac14(1-a)\overline{w} & \frac14(1-b)\overline{z} &(1-a)(1-b)\end{array}\right]\ .
\end{equation*}

Note  that $\tr_2$, the partial trace over the second factor, is obtained by adding up the two diagonal $2\times 2$
blocks, and $\tr_1$, the partial trace over the first factor, is obtained by taking the trace in each $2\times 2$ block. 
Therefore,
\begin{equation}\label{ncwild}
\rho_1\star \rho_2 =  \left[\begin{array}{cc} \frac12(a+b) & \frac{z}{8}(2-b)\\ \frac{\overline{z}}{8}(2-b) & 1 -\frac12(a+b)\end{array}\right]\ .
\end{equation}
In particular, taking $\rho = \rho_1$,
\begin{equation}\label{ncwild2}
\rho\star \rho=  \left[\begin{array}{cc} a & \frac{z}{8}(2-a)\\ \frac{\overline{z}}{8}(2-a) & 1 -a\end{array}\right]\ ,
\end{equation}
which is nonlinear in $\rho$. Note also, that in contrast with the classical case, the quantum Wild convolution is not commutative; $\rho_1\star \rho_2 \neq \rho_2\star \rho_1$ when $z(2-b) \neq w(2-a)$. 
\end{exam} 

\begin{exam}\label{deg} Take $\H = \C^3$,  so that $\H_N = (\C^3)^{\otimes N}$.   Define the single particle Hamiltonian $h$ by 
  $h = \left[\begin{array}{ccc} 0 & 0 &0 \\ 0 & 1 & 0\\ 0 & 0 &2 \end{array}\right]$ so that the $N$-particle Hamiltonian
 $H_N = \sum_{j=1}^N h_j$ has $\spec   =\{0,\dots,2N\}$.   In particular ${\rm Spec}(H_2) = \{ 0,1,2,3,4\}$.
Let $\{\psi_1,\psi_2,\psi_3\}$ be the standard basis of $\C^3$ so that $h\psi_j = (j-1)\psi_j$. Then the eigenspace of
$H_2$ with eigenvalue $2$ is spanned by 
$$\psi_1\otimes \psi_3\ ,\quad  \psi_2\otimes \psi_2\ \quad  {\rm and} 
\quad \psi_3\otimes \psi_1\ .$$
Therefore,
$$\tr_2[\sigma_2] = \frac13 (|\psi_1 \bk \psi_1| + |\psi_2 \bk \psi_2| + |\psi_3 \bk \psi_3|)\ .$$
Take $\aC$ to be the subgroup of $U(9)$ that commutes with $H_2$ considered as an operator on $\C^9$.
Let $U$ be the identity map on $\aC$, and let $\nu$ be the uniform Haar measure on $\cC$. 
Using Lemma~\ref{explicit} it is easy to see that $\rho \mapsto \rho \star \rho$ is non-linear. 
\end{exam}

\subsection{Steady states for the  QKBE}  

Through this subsection we fix an ergodic collision specification $(\aC,U,\nu)$, and let 
$\star$ denote the corresponding Wild convolution.  Let
\begin{equation}\label{hspecres}
h = \sum_{e\in {\rm Spec}(h)} e P_e
\end{equation}
be a spectral resolution of the single particle Hamiltonian $h$.

The steady state solutions of the QKBE are precisely the $\rho\in \Dens(\H)$ such that $\rho = \rho\star \rho$.
The Gibbs states $\rho_\beta = Z_\beta^{-1} e^{-\beta h}$  are always steady states since 
$$\rho_{\beta}\otimes \rho_\beta = Z_{\beta}^{-2} e^{-\beta H_2} \in \cA_2$$
and thus $U(\sigma) \rho_{\beta}\otimes \rho_\beta  U^*(\sigma) = 
\rho_{\beta}\otimes \rho_\beta $ for all $\sigma\in \aC$. It follows that
$\cQ(\rho_{\beta}\otimes \rho_\beta) = \rho_{\beta}\otimes \rho_\beta$, 
and then that $\rho_{\beta}\star\rho_\beta  =\rho_\beta$. 
Whether or not there are other steady states depends on the spectrum of $h$ in a way that will be specified below. 
For any density matrix $\rho$, $S(\rho) = -\tr[rho \log \rho]$ is the von Neumann entropy of $\rho$.  The set of finite entropy steady states turns out to be independent of the particular ergodic collision specification $(\aC,U \nu)$, but depends only on $h$.

\begin{thm}\label{steadystate}
Let $h$ have the spectral resolution (\ref{hspecres}), and let $\rho \in \Dens(\H)$ be such that $\rho = \rho\star \rho$ and $S(\rho) < \infty$. Then $\rho$ has the form
\begin{equation}\label{rhospecres}
\rho = \sum_{e\in {\rm Spec}(h)} \lambda_e P_e
\end{equation}
for non-negative numbers $\{\lambda_e\ :\ e\in {\rm Spec}(h)\}$ such that $\sum_{e\in{\rm Spec}(h)} \tr[P_e]\lambda_e =1$.
Moreover, if $\{e_i,e_j,e_k,e_\ell\} \subset {\rm Spec}(h)$  then
\begin{equation}\label{rhospecres2}
e_i+e_j = e_k+e_\ell \quad \Rightarrow \quad \log  \lambda_{e_i}+\log  \lambda_{e_j} = \log  \lambda_{e_k}+\log  \lambda_{e_\ell}\ .
\end{equation}
Conversely, every such $\rho\in \Dens(\H)$ is a steady state. 
\end{thm} 

\begin{proof}  
Recall that by condition {\it (iv)} in the definition of a collision specification, $\cQ$ commutes with the map $X \mapsto VXV^*$ 
where $V$ is the swap transformation on
$\H_2$ given by $V \phi\otimes \psi = \psi\otimes \phi$ for all $\phi,\psi\in \H$. Therefore, for all bounded $A$ on $\H$,
\begin{multline*}
\tr[A\otimes \one_\H \cQ(\rho \otimes \rho)] = 
\tr[V  \one_\H \otimes A  V^*  \cQ(\rho \otimes \rho)] =\\ \tr[  \one_\H \otimes A   \cQ(V \rho \otimes \rho V^*)]  = 
\tr[  \one_\H \otimes A   \cQ( \rho \otimes \rho)] 
\end{multline*}
Consequently, 
\begin{equation}\label{twoway}
\rho\star \rho = \tr_2[\cQ(\rho\otimes \rho)] =  \tr_1[\cQ(\rho\otimes \rho)] \ .
\end{equation}
By the subadditivtiy of the entropy,
\begin{equation}\label{twowayB}
2S(\rho\star\rho) \geq S(\cQ(\rho\otimes \rho))\ , 
\end{equation}
and then by \eqref{twoway} and the  concavity of the von Neumann entropy
\begin{eqnarray}
2S(\rho\star \rho)  &\geq & -\tr\left[ \cQ (\rho\otimes \rho) \log (\cQ (\rho\otimes \rho))\right]\nonumber \\
&\geq& -\tr\left[ (\rho\otimes \rho) \log (\rho\otimes \rho)\right] = 2S(\rho)\ .
\end{eqnarray}
Since the von Neuman entropy is {\em strictly concave}, there is equality  if and only if 
$\sigma \mapsto U(\sigma) \rho \otimes \rho U^*(\sigma)$ is constant almost everywhere with respect to $\nu$, and then by the continuity of $\sigma \mapsto U(\sigma)$ and the fact that $u(\sigma_0) = \one_{\H_2}$, this means that
$$\sigma \mapsto U(\sigma) \rho \otimes \rho U^*(\sigma) = \rho\otimes \rho$$
for all $\sigma$. By the ergodicity, this means that $\rho\otimes \rho\in \cA_2$. 

Therefore, when $\rho \star \rho = \rho$, it is also the case that $\rho \widehat{\star}\rho = \rho$ where
$\rho \widehat{\star}\rho  = \tr_2[ {\rm E}_{\cA_2}(\rho \otimes \rho)]$ 
is the Wild convolution corresponding to the uniform average over all of the unitaries commuting with 
$H_2$. In this case, Lemma~\ref{nondeg} applies and 
$$\rho \otimes \rho = \rho \widehat{\star}\rho   = \cD_h(\rho)\ .$$
Hence if $\{\psi_j\}_{j\in \cJ}$ is an orthonormal basis of $\H$ consisting of eigenfunctions of $h$, then for some sequence
$\{\mu_j\}_{j\in \cJ}$, 
$$\rho = \sum_{j\in \cJ} \mu_j |\psi_j \bk \psi_j|\ .$$
It then follows that 
$$\rho\otimes \rho =  \sum_{j,k\in \cJ} \mu_j\mu_k  |\psi_j\otimes \psi_k \bk \psi_j\otimes \psi_k|\ .$$
for the right hand side to belong to $\cA_2$, it is necessary and sufficient that whenever $e_j+e_k = e_\ell + e_m$, then 
$\mu_j\mu_k = \mu_\ell \mu_m$. Taking $m=k$ for some $k$ such that $\mu_k \neq 0$, we see that $e_j = e_\ell$ implies that
$\mu_j = \mu_\ell$, and thus $\rho$ has the expansion in the form (\ref{rhospecres}) and then by the same reasoning once more we obtain (\ref{rhospecres2})
\end{proof}

Theorem~\ref{steadystate} says in particular that if $\rho$ is a steady state solution of the QKBE for an ergodic collision specification, then $\rho = f(h)$ for some real valued function on ${\rm Spec}(h)$.   This may be the only restriction. If $h$ is such that
whenever $e_j+e_k = e_\ell + e_m$ then either $e_j = e_\ell$ and $e_k = e_m$ or else $e_j = e_m$ and $e_k = e_\ell$, then
(\ref{rhospecres2}) imposes no restriction, and indeed, we have seen that in this case, if $\rho = f(h)$, so that $\rho = \cD_h(\rho)$,
then $\rho\star \rho - \rho$. 

On the other hand, suppose $h$ has evenly spaced eigenvalues and there are at least three of them. To be specific, suppose that 
${\rm dim}(\H) = n\geq 3$, and ${\rm Spec}(h) = \{ 0, 1,\dots,n-1\}$. Then for each $j=1,\dots n-2$, $e_{j-1}+ e_{j+1} = 2e_j$, and hence $\lambda_{e_j} = \sqrt{ \lambda_{e_{j-1}}\lambda_{e_{j+1}}}$. This means that for some $\beta\in \R$,
$\rho = Z_\beta^{-1}e^{-\beta h}$. (In finite dimension, negative temperatures are allowed.) In general, the more ways a given eigenvalue $E$ of $H_2$ can be written as a sum of eigenvalues of $h$, the more constraints there are on the set of steady state solutions of the QKBE.

\begin{defi}[Steady states and collision invariants] Let $h$ be a self adjoint operator on $\H$ that is bounded below hand has a compact resolvent. The
set $\Dens_{\infty,h}(\H)$ consists of those $\rho \in \Dens(\H)$ such  that (\ref{rhospecres}) and (\ref{rhospecres2}) are satisfied.
The set $\Dens_{\infty,h}(\H)^\circ$ consist of those $\rho \in \Dens_{\infty,h}(\H)$ that are strictly positive. The set of {\em collision invariants} is the set of self adjoint operators $A$ of the form $A = \log \rho$, $\rho \in \Dens_{\infty,h}(\H)^\circ$.
\end{defi}

The term ``collision invariant'' is justified by the next theorem.

\begin{thm}\label{colinvthm}  Let $\rho_\infty \in \Dens_{\infty,h}(\H)^\circ$. Then for all $\rho\in \Dens(\H)$,
\begin{equation}\label{conserve}
\tr[ \log(\rho_\infty) \rho]  =   \tr[ \log(\rho_\infty) \rho\star \rho]\ .
\end{equation}
In particular, for every solution $\rho(t)$ of the QKBE,  and every collision invariant $A$, $\tr[A \rho(t)]$ is independent of $t$. Moreover, for each $\rho_\infty \in \Dens_{\infty,h}(\H)$ the relative entropy $D(\rho(t)\|\rho_\infty)$ is strictly monotone decreasing along any solution that is not a steady state solution. 
\end{thm}

\begin{proof} Recall that as a consequence of condition {\it (iv}) in the definition of a collision specification, we always have
(\ref{twoway}). Therefore,
$$2 \tr[\log (\rho_\infty)\rho\star \rho ] = \tr[\log(\rho_\infty \otimes \rho_\infty) \cQ( \rho \otimes \rho)] = 
\tr[\cQ(\log(\rho_\infty \otimes \rho_\infty)  )\rho \otimes \rho]\ .$$
However, since $\rho_\infty \otimes \rho_\infty \in cA_2$, $\log(\rho_\infty \otimes \rho_\infty+ \epsilon \one_{\H_2}))\in \cA_2$ 
for all $\epsilon >0$ (this is only necessary if $\rho_\infty$ has zero as an eigenvalue) so that 
$$\cQ(\log(\rho_\infty \otimes \rho_\infty + \epsilon \one_{\H_2}))  = \log(\rho_\infty \otimes \rho_\infty + \epsilon \one_{\H_2})\ .$$
Therefore,
$$2 \tr[\log (\rho_\infty)\rho\star \rho ] = \tr[\log(\rho_\infty \otimes \rho_\infty) \rho \otimes \rho] = 2 \tr[\log(\rho_\infty)\rho]\ .$$

Now note that 
$$D(\rho(t)\|\rho_\infty)  = \tr[\log(\rho(t))\rho(t)] - \tr[\log(\rho_\infty)\rho(t)] \ ,$$
and we have already seen that due to the strict convexity of $t\mapsto t \ln t$, $t \mapsto \tr[\log(\rho(t))\rho(t)]$ is strictly decreasing at each $t$ unless $\rho(t)$ is a steady state. Finally, from the Wild sum representation of $\rho(t)$, it is clear that unless $\rho(0)$ is a steady state, $\rho(t)$ is not a steady state for any finite $t$. 
\end{proof}

\begin{remark}  When $\Dens_{\infty,h}(\H)$ consists only of the Gibbs states 
$Z_\beta e^{-\beta h}$, then Theorem~\ref{colinvthm}
provides only one conservation law, namely that $\tr[h \rho(t)]$ is constant so that the energy is conserved. This is the familiar situation with the classical Kac-Boltzmann equation. However, we have seen that if every $E\in {\rm Spec}(H_2)$ is the sum of a single (unordered) pair of eigenvalues of $h$, then  $\Dens_{\infty,h}(\H)$ consists of all $\rho \in \Dens_h(\H)$ such that $\rho$ commutes with $h$. This means that the diagonal entries of $\rho$ in an eigenbasis for $h$ are conserved, as we have seen.
\end{remark}

We are now ready to study the basins of attraction of the steady states. Fix some $\rho_\infty \in \Dens_{\infty,h}(\H)$. Let $\rho(t)$ be the solution with $\rho(0) = \rho_0\in \Dens(\H)$. For $\lim_{t\to\infty}\rho(t) = \rho_\infty$
so be valid in the topology making all of these functionals continuous, we require that for all $\widehat \rho_\infty \in 
\Dens_{\infty,h}(\H)^\circ$ (so that 
$\widehat \rho_\infty $ has no zero eigenvalues), 
$$\tr[\log (\widehat \rho_\infty ) \rho_0] = \tr[\log (\widehat \rho_\infty ) \rho_\infty] \ .$$
We are interested in conditions under which this is also a sufficient condition.  

\subsection{The linearized QKBE}

We begin the investigation of the long-time behavior of solutions of the QKBE  by linearizing it in the vicinity of a steady state. This has to be done with some care: To obtain a purely dissipative  linear equation on a Hilbert space, we must choose the inner product to reflect some dissipative feature of the non-linear equation. This means the inner product must ultimately derive from the dissipativity of relative entropy for the QKBE. 

We briefly recall the  linearization of the classical KBE. Let $M(v) = (2\pi)^{-1/2}e^{-|v|^2/2}$ be the steady state about which we shall linearize. Let $\rho(v)$ be a probability density on $\R$ such that 
\begin{equation}\label{classicalex1}
\int_{\R} \rho(v) v^2{\rm d}v = \int_{\R} M(v) v^2{\rm d}v =1\ ,
\end{equation}
so that $\rho$ has the same conserved energy as the steady state $M$.  Let $\rho(t)$ denote the solution of the classical KBE with initial data $\rho$. Suppose that $S(\rho)$ is finite. 
(In the classical case, this means $S(\rho) > -\infty$ since $S(\rho) \leq S(M)$.)

We now write 
\begin{equation}\label{classicalex2}
\rho = M(1+f)
\end{equation}
where ${\displaystyle \int_{\R} v^2 f(v)M(v){\rm d}v =0}$.  Assuming that $f$ is ``small'', we make the expansion
$$S(\rho) = S(M(1+f)) = -\int_{\R}[\log M + \log (1+f)] M(1+f){\rm d}v = S(M)  - \frac12 \int_{\R}f^2 M{\rm d}v \ .$$
The Hilbert space we use to linearize the Kac-Boltzmann equations then is $L^2(\R,M(v){\rm d}v)$ and the perturbation $\rho$ of $M$ is written in the form $\rho = M(1+f)$, with $f\in L^2(\R,M(v){\rm d}v)$ because then for $f$ small in this Hilbert space,
$\tfrac12 \|f\|_{L^2(\R,M(v){\rm d}v)}$ is the second order approximation of $D(M(1+f)\|M)$.  Using this scheme to linearize the KBE yields a purely dissipative linear equation in $L^2(\R,M(v){\rm d}v)$ because of the close connection between the 
Hilbert space $L^2(\R,M(v){\rm d}v)$ and the Hessian of $S(\rho)$ and hence $D(\rho\|M)$. 

We seek to follow this model in the quantum case, but we must  take into account that due to non-commutativity, there are many natural analogs of $L^2(\R,M(v){\rm d}v)$ when we replace $M$ be a density matrix  on $\H$. Consideration of the entropy leads, as above, to the useful analog.

For the rest of this section, to postpone technical difficulties, we suppose that $\H$ is finite dimensional. 
If $\rho\in \Dens(\H)$ is strictly positive, and $A$ is self-adjoint in $\cB(\cH)$, and then if $\tr[A] =0$,  $\rho+tA \in \Dens(\H)$ for all $|t|$ sufficiently small. Then
$$\frac{{\rm d}}{{\rm d}t} S(\rho+ tA)\bigg|_{t=0} = \tr[ \log (\rho) A]\ .$$
Moreover,
$$\frac{{\rm d}}{{\rm d}t} \log( \rho + tA) \bigg|_{t=0} = \int_0^\infty \frac{1}{s\one_\H + \rho} A  \frac{1}{s\one_\H + \rho}{\rm d}s\ .$$

For any positive operator $B \in \cB(\H)$, define the linear map $[B]^{-1}:\cB(\H) \to \cB(\H)$ by
$$[B]^{-1}A = \int_0^\infty \frac{1}{s\one_\H + B} A \frac{1}{s\one_\H + B} \ .$$
This is a non-commutative version of ``division by $B$''.  The inverse operation, $[B]:\cB(\H) \to \cB(\H)$
is given by
$$[B]A = \int_0^1 B^s A B^{1-s} {\rm d} s\ ,$$
and this is a non-commutative version of ``multiplication by $B$''. 

The computations made above show that the Hessian of $\rho \mapsto S(\rho)$ at $\rho$ is given by the quadratic form
$\tr[ A [\rho]^{-1} A]$. That is, with $\rho$ and $A$ as above,
$$\frac{{\rm d}^2}{{\rm d}t^2} S(\rho+ tA)\bigg|_{t=0} = \tr[ A [\rho]^{-1} A]\ .$$

The {\em Bogulioubov-Kubo-Mori} inner product on $\cB(\H)$ with reference state $\rho\in \Dens(\H)$ is the inner product
$\langle \cdot ,\cdot\rangle_{BKM}$ given by
$$\langle A ,B\rangle_{BKM} = \tr[A^* [\rho]B]\ .$$

We are now ready to linearize.  Fix some strictly positive steady state $\rho_\infty$.  Choose some $\rho$ in the possible basin of attraction of $\rho_\infty$ that is ``close'' to $\rho_\infty$.  That is, for all strictly positive steady states $\widehat \rho_\infty$,
$$\tr[ \log(\widehat \rho_\infty)\rho] = \tr[ \log(\widehat \rho_\infty)\rho_\infty] \ .$$
This is the quantum analog of (\ref{classicalex1}).

Define a self-adjoint operator $A$ by $A = [\rho_\infty]^{-1}(\rho - \rho_\infty)$ so that 
\begin{equation}\label{linearize2}
\rho = [\rho_\infty](1 + A)\ . 
\end{equation}
This is the direct analog of (\ref{classicalex2}). 

We now apply this to $\rho(t) =  [\rho_\infty](1 + A(t))$ and discard the terms that are quadratic in $A(t)$ in the QKBE. We obtain:
\begin{equation}
\frac{{\rm d}}{{\rm d}t}A(t) = 
2\left(  [\rho_\infty]^{-1} [ \rho_\infty \star ([\rho_\infty] A(t)) + ([\rho_\infty] A(t))  \star  \rho_\infty ] - [\rho_\infty]A(t)\right)
\ .
\end{equation}

\begin{defi} For a strictly positive steady state $\rho_\infty$,  the {\em linearized QKBE} operator is the operator $\aK$ on $\cB(\H)$ defined by
$$\aK  X = 2\left( [\rho_\infty]^{-1}[ \rho_\infty \star X  + X  \star  \rho_\infty ]  - X\right)\ .$$
The linearized QKBE at $\rho_\infty$ is the equation
$$
\frac{{\rm d}}{{\rm d}t} X(t) = \aK X(t)\ .
$$
\end{defi}

\begin{thm}\label{collinv}  Let $\aK$ be the linearized Kac-Boltzmann operator at a steady state $\rho_\infty$. Let $\langle \cdot, \cdot\rangle_{BKM}$ be the corresponding inner product on $\cB(\H)$. Then for all $A,B\in \cB(\H)$,
\begin{equation}\label{lindis}
\langle B, \aK  A\rangle_{BKM} = \langle \aK B,  A\rangle_{BKM} \quad{\rm and} \quad   \langle  A, \aK A\rangle_{BKM} \leq 0\ .
\end{equation}
Moreover $ \langle  A, \aK A\rangle_{BKM} =0$ if and only if $A$ is in the linear span of the collision invariants. 
\end{thm}

\begin{proof}
From the definition,
\begin{equation}\label{linearize3}
\langle B, \aK A\rangle_{BKM} = 2\tr[ B^*(  \rho_\infty \star ([\rho_\infty] A) + ([\rho_\infty] A)  \star  \rho_\infty ]] - \tr[B^*[\rho]A]\ .
\end{equation}
We now compute
\begin{eqnarray}\label{linearize5}
\tr_\H[ B^*(  ([\rho_\infty] A)\star  \rho_\infty ] 
&=&\int_{0}^1  \tr_{\H_2}[ (B\otimes \one_\H)(\rho_\infty^s A \rho_\infty^{1-s}\star \rho_\infty)]{\rm d}s \nonumber\\
&=&\int_{0}^1  \tr_{\H_2}[ (B\otimes \one_\H)\cQ(\rho_\infty^s A \rho_\infty^{1-s}\otimes  \rho_\infty)]{\rm d}s \nonumber\\
&=&\int_{0}^1  \tr_{\H_2}[ (B\otimes \one_\H)\cQ((\rho_\infty\otimes \rho_\infty)^s A\otimes \one_{\H} 
(\rho_\infty\otimes \rho_\infty)^{1-s}]{\rm d}s \nonumber
\end{eqnarray}
since each $U(\sigma)$ commutes with $\rho_\infty \otimes \rho_\infty$, it follows that
$$\tr_\H[ B^*(  ([\rho_\infty] A)\star  \rho_\infty ]  = 
\tr_\H\left[B^* \tr_2\left[\int_{0}^1(\rho_\infty\otimes \rho_\infty)^s  \cQ(A \otimes \one_\H)
(\rho_\infty\otimes \rho_\infty)^{1-s}\right]{\rm d}s\right]\ .$$
A similar computation shows that
$$\tr_\H[ B^*(  \rho_\infty \star ([\rho_\infty] A)]  = 
\tr_\H\left[B^* \tr_2\left[\int_{0}^1(\rho_\infty\otimes \rho_\infty)^s  \cQ( \one_\H \otimes A)
(\rho_\infty\otimes \rho_\infty)^{1-s}\right]{\rm d}s\right]\ .$$
This gives us an alternate expression for $\aK$:
\begin{equation}\label{alternate}
\aK A = [\rho_\infty]^{-1} \left( [\rho_\infty \otimes \rho_\infty] \cQ(A \otimes \one_\H +  \one_\H \otimes A)\right) - A\ .
\end{equation}
It is now easy to see from the computations above that
\begin{equation}\label{linearize4}
\langle B, \aK A\rangle_{BKM} = \langle  \aK B, A\rangle_{BKM}\  .
\end{equation}
To display the fact that  
\begin{equation}\label{linearize5}
\langle A, \aK A\rangle_{BKM} \leq 0\  ,
\end{equation}
and to identify the null space of $\aK$, it is useful to express $\aK$ in yet one more form in which the the constituents of the operator $\cQ$ are written out explicitly which permits further symmetrization.  Taking advantage of the fact that 
$$[\rho_\infty \otimes \rho_\infty](A \otimes \one_\H + \one_H\otimes A) = [\rho_\infty]A\otimes \rho_\infty + 
\rho_\infty \otimes   [\rho_\infty]A\ ,$$
we have
$$\tr_2\left[  [\rho_\infty \otimes \rho_\infty](A \otimes \one_\H + \one_H\otimes A)\right] = [\rho_\infty]A$$
since $\tr[[\rho_\infty]A]=0$.
Therefore, we can rewrite (\ref{alternate}) as
\begin{multline}\label{alternate2}
[\rho_\infty] \aK A = \\
\tr_2\left[ \int_{\aC}{\rm d}\nu  [\rho_\infty \otimes \rho_\infty] \left(U(\sigma)[A \otimes \one_\H +  \one_\H \otimes A] U^*(\sigma) - [A \otimes \one_\H +  \one_\H \otimes A]\right)\right] \ .
\end{multline}
To shorten the expression that follow, we temporarily introduce the notation ${\mathcal A} = A \otimes \one_\H +  \one_\H \otimes A$
and ${\mathcal B} = B \otimes \one_\H +  \one_\H \otimes B$ for $A,B\in \cB(\H)$. 
Then since $\langle B,\aK A\rangle_{BKM} = \tr [B^* [\rho_\infty] \aK A]$, we have that 
\begin{eqnarray}
\langle B,\aK A\rangle_{BKM} &=& \tr_{\H_2}\left[ B^*\otimes \one_\H \int_{\aC}{\rm d}\nu  [\rho_\infty \otimes \rho_\infty] \left(U(\sigma){\mathcal A} U^*(\sigma) - {\mathcal A}\right)\right]\nonumber\\
&=& \frac12  \int_{\aC}{\rm d}\nu \tr_{\H_2}\left[  {\mathcal B}^* [\rho_\infty \otimes \rho_\infty] \left(U(\sigma){\mathcal A} U^*(\sigma) - {\mathcal A}\right)\right]\nonumber\\
&=& \frac12  \int_{\aC}{\rm d}\nu \tr_{\H_2}\left[ U^*(\sigma) {\mathcal B}^* [\rho_\infty \otimes \rho_\infty] \left(U(\sigma){\mathcal A} U^*(\sigma) - {\mathcal A}\right)U(\sigma)\right]\nonumber\\
&=& \frac12  \int_{\aC}{\rm d}\nu \tr_{\H_2}\left[ U^*(\sigma) {\mathcal B}^*U(\sigma)U^*(\sigma) [\rho_\infty \otimes \rho_\infty] \left(U(\sigma){\mathcal A} U^*(\sigma) - {\mathcal A}\right)U(\sigma)\right]\nonumber\\
&=& \frac12  \int_{\aC}{\rm d}\nu \tr_{\H_2}\left[ (U^*(\sigma) {\mathcal B}U(\sigma))^* [\rho_\infty \otimes \rho_\infty] \left({\mathcal A}  - U^*(\sigma){\mathcal A}U(\sigma)\right)\right]\nonumber\\
&=& \frac12  \int_{\aC}{\rm d}\nu \tr_{\H_2}\left[ (U(\sigma) {\mathcal B}U^*(\sigma))^* [\rho_\infty \otimes \rho_\infty] \left({\mathcal A}  - U(\sigma){\mathcal A}U^*(\sigma)\right)\right]\ .\nonumber
\end{eqnarray}
The second equality is from the invariance under the swap, the third equality  is the unitary invariance of the trace on $\H_2$,
the fourth equality is trivial, the fifth equality  is that $\rho_\infty\otimes \rho_\infty$  commutes with each $U(\sigma)$, and the 
sixth equality is the invariance under the adjoint.   Now averaging the expressions in the second and sixth lines, we have
\begin{multline}\label{comsq}
\langle B,\aK A\rangle_{BKM}  = \\ 
-\frac14 \int_{\aC}{\rm d}\nu \tr_{\H_2}\left[ (\mathcal{B} - U(\sigma) {\mathcal B}U^*(\sigma))^* [\rho_\infty \otimes \rho_\infty] \left({\mathcal A}  - U(\sigma){\mathcal A}U^*(\sigma)\right)\right]\ .
\end{multline}
From this expression it immediately clear that $\langle A,\aK A\rangle_{BKM}  \leq 0$, and there is equality if and only if 
${\mathcal A}$ commutes with each $U(\sigma)$, and hence that ${\mathcal A}$ is a collision invariant. 
\end{proof}

The is a simpler proof of the fact that $ \langle  A, \aK A\rangle_{BKM} \leq 0$ based directly on entropy dissipation. 
However, this is a consequences of the monotonicity of the entropy under the nonlinear QKBE.

To see this, consider a self-adjoint $A$ such that 
$$\langle \log(\widehat \rho_\infty), A\rangle_{BKM} = 0$$
for all strictly positive steady states $\widehat \rho_\infty$. Then for all $u$ with $|u|$ sufficiently small,
$$\rho_u= [\rho_\infty](\one_\H + u A) = \rho_\infty + u [\rho_\infty]A$$
belongs to $\Dens(\H)$ and $\tr[\log(\widehat \rho_\infty)\rho_u] = \tr[\log(\widehat \rho_\infty)\rho_\infty]$
for all strictly positive steady states $\widehat \rho_\infty$.
Let $\rho_u(t)$ be the solution of the QKBE with initial data $\rho_u$. By the entropy production inequality,
for all $u$ with $|u|$ sufficiently small,
$$
\frac{{\rm d}}{{\rm d}t}S(\rho_u(t))\bigg|_{t=0} = -2\tr[\log \rho_u (\rho_u\star \rho_u - \rho_u)] \geq 0\ .
$$
That is, for all $u$ with $|u|$ sufficiently small,
$$
\tr[\log \rho_u (\rho_u\star \rho_u)] \leq  \tr[ \rho_u\log( \rho_u)]
$$
Define $f(u) = \tr[\log \rho_u (\rho_u\star \rho_u)]$ and $g(u) = \tr[ \rho_u\log( \rho_u)]$. Evidently $f(0)=  g(0)$. Also,
$$f'(0) = g'(0) = \tr[\log(\rho_\infty)[\rho_\infty] A] = 0\ .$$
It follows that $f''(0) \leq g''(0)$, and doing the computation, this proves (\ref{linearize5}).
However, this approach does not seem to characterize the null space of $\aK$. 

\begin{defi}[Spectral gap of the linearized QKBE]  The number $\Delta_{KB}$ given by
$$\Delta_{KB} = \inf \left\{ \frac{  \langle  A, \aK A\rangle_{BKM}}{ \langle  A,  A\rangle_{BKM}}\ :\   \langle  B, A\rangle_{BKM} =0 \ {\rm for \ all\ collisions\ invariants\ } B\right\}\ .$$

\end{defi}

At this point it is not clear that in general $\Delta_{KB} > 0$ or that $\Delta_{KB}$ is independent of $\rho_\infty$. (Both things are true for the simplest model with $\H = \C^2$ discussed above, we we shall see.)

\subsection{Lyapunov functionals}  

We briefly return to the the subject of convergence to equilibrium for the QKBE.  
McKean proved \cite{McK} a theorem for the classical Kac Boltzmann Equation that explains the special role of the entropy in this theory.  
He showed that the only 
functionals on probability densities on $\R$ of the form 
$\rho \mapsto \int_\R \Phi(\rho(v)){\rm d}v$, with 
$\Phi:(0,\infty)\to \R$,  that are montone for every 
  solution $\rho(v,t)$ of Kac's equation (on which the functional is defined),   are those in which $\Phi$ has the from $\Phi(t) = 
  at \log t + bt$ for some constants $a$ and $b$.  
  
We close this paper by using the results obtained in this section to prove a quantum analog in 
the case in which for each energy $E$ there is a unique invariant state, as in the case considered by McKean. 

\begin{thm}  Let $\Phi:(0,t)\to \R$, and suppose that $\tr[\Phi(\rho(t)]$ is monotone for every solution of the QKBE for some collision specification such that at each energy $E$  there is exactly one invariant state $\rho_\infty$. Then $\Phi$ necessarily has the form 
$\Phi(t) = a t\log  t + bt$ for some constants $a$ and $b$. 
\end{thm}

\begin{proof}  Let $\Phi$ be such a function.  Let 
$\rho_\infty$ be the steady state of the QKBE at a given energy $E$. We may assume that $\tr[\Phi(\rho(t)]$ is monotone decreasing along solutions.

For $\epsilon>0$ let $\rho := [\rho_\infty](1 + \epsilon A) $ be a density matrix that is a small perturbation of $\rho_0$. 
The montonicity implies that
$$\langle  \Phi'([\rho_\infty](1 + \epsilon A)),\aK A\rangle_{BKM} \leq 0 \ .  $$
By continuity it follows that 
$$\langle  \Phi'(\rho_\infty),\aK A\rangle_{BKM} \leq 0 \ ,  $$
and then by Theorem~\ref{collinv}, $\Phi'(\rho_\infty)$ must be a collision invariant.  If $\rho_\infty$ is the only steady state at a given energy $E$, then the equation
$$\Phi'(\rho_\infty) = a \log \rho_\infty + b$$
must be valid for some $a$. As one varies $E$, the eigenvalues of $\rho_\infty$ vary and one must have
$$\Phi'(x) = a \log x + b$$
meaning that $\Phi(x) = a x\log x  + (b-a)x$.  

Hence when there is a unique invariant state for each energy $E$, entropy is the only non-trivial Lyapunov function, exactly as in the classical case examoned by McKean.
\end{proof}

The investigation completed in this paper sets the stage for the investigation of the rates of approach to equilibrium for the QKME, the QBBE, and the linearized QKBE, and the quantitative relations between these rates. The theorem just proved explains why, just as in the classical case, relative entropy will play an important role in this investigation. In the classical setting, Cercignani \cite{Cer}
had conjectured that an inequality bounding the entropy dissipation from below by a constant multiple of the relative entropy would be valid; see \cite{CCL3}. 
This turns out not to be true, either at the level of the Kac-Boltzmann equation \cite{BC} or the Kac Master equation \cite{CCRLV,Einav1} -- though it is valid for some non-physical collision models, and, surprisingly, this can be used to study physical models,
 and in physical models, Cercignani's conjecture is {\em almost} true  \cite{V}. The  known counter-examples use  states in which a very large fraction of the energy is concentrated in a very small number of particles. 
 
 Such states may be easier to rule out   in the quantum setting:  Already in the discussion follwing Corollary~\ref{sepcl}, we have introduced the quantum analog of Cercignani's conjecture.  Certainly in models in which the single particle state space is finite dimensional,  states in which most of the energy is concentrated in a small fraction of the partcles cannot   exist. Quantum entropy production  will be investigated in forthcoming work. The present investigation provides a fairly detailed and complete description of the possible steady states, which sets the stage for this, although even here interesting questions remain open. For example, can one classify the functions $\Phi$ yielding monotone functionals as in the last theorem when the class of collision invariants, and steady states, is large?

 \bigskip

\noindent{\bf Acknowledgements}  Work of E.C was  partially supported by U.S. National Science Foundation
grant DMS 1501007.  Work of M.C. was  partially supported by UID/MAT/04561/2013 and SFRH/BSAB/113685/2015. Work of M.L. 
was partially supported by U.S. National Science Foundation
grant DMS 1600560.   Substantial progress on this work was made while all three authors were at the Mittag-Leffler Institute in the Fall of 2016, and we are grateful for the hospitality and environment there that made this possible.

\end{document}